\documentclass[11pt,a4paper]{article}
\usepackage{mathpazo}
\usepackage[hyphens]{url}
\usepackage{graphicx,natbib,paralist,amsmath,bm,mathtools,todonotes,booktabs,microtype,hyperref,geometry,setspace}
\usepackage{mathrsfs}
\usepackage[no-weekday]{eukdate}
\usepackage[ruled,vlined,linesnumbered]{algorithm2e}
\mathtoolsset{showonlyrefs}
\usepackage[labelfont=bf, textfont=it, labelsep=colon, format=hang, singlelinecheck=true]{caption}
\usepackage{titling}
\pretitle{\begin{center}\LARGE\bfseries}
\posttitle{\end{center}\vspace*{0.3cm}\par}
\usepackage{amsthm}

\newtheorem{lemma}{Lemma}[section]
\newtheorem{proposition}{Proposition}[section]
\theoremstyle{definition}
\newtheorem{definition}{Definition}[section]

\allowdisplaybreaks
\newcommand{\blind}{0}

\graphicspath{{Figures/}}

\begin{document}

\title{Anomaly detection in dynamic networks}

\if0\blind
	{
		\author{
			Sevvandi Kandanaarachchi$^1$\\
			Data61, CSIRO\\
			and \\
			Rob J Hyndman \\
			Monash University\\
		}
	} \fi

\maketitle

\footnotetext[1]{sevvandi.kandanaarachchi@data61.csiro.au}

\begin{abstract}
	Detecting anomalies from a series of temporal networks has many applications, including road accidents in transport networks and suspicious events in social networks. While there are many methods for network anomaly detection, statistical methods are under utilised in this space even though they have a long history and proven capability in handling temporal dependencies. In this paper, we introduce \textit{oddnet}, a feature-based network anomaly detection method that uses time series methods to model temporal dependencies. We demonstrate the effectiveness of oddnet on synthetic and real-world datasets. The R package \texttt{oddnet} implements this algorithm.
\end{abstract}

\noindent%
\vfill
\newpage
\spacing{1.2}

\section{Introduction}\label{introduction}

Networks are ubiquitous, including biological networks such as neurological systems in our brains, infrastructure networks such as road, train, flight and electricity networks, and social networks including those provided by many social media platforms. Almost all networks have a time component: the traffic in our road networks change over time, different neurons fire and wire together as we do diverse tasks, and friend networks rarely remain constant. Detecting anomalies in dynamic networks is important because an anomaly may signify a fault, a road accident, a suspicious event, or an unusual incident that needs investigation. While statistical methods have a long history and success in modelling temporal dependencies, their use in network anomaly detection has been limited \citep{Ranshous2015}. In this paper, we present a network anomaly detection method that we call \textit{oddnet}, which fuses network analysis with time series techniques to identify anomalies in dynamic networks.

Different types of anomalies are studied in the context of dynamic networks; anomalous nodes/vertices, anomalous edges, and anomalous subgraphs are some structures of interest. Our focus is on anomalous networks; i.e., we find time points at which the network is anomalous. \citet{Akoglu2015} broadly categorise anomaly detection techniques in dynamic networks into four groups: feature-based, decomposition-based, community or clustering-based, and window-based methods. In feature-based systems, a ``good summary'' of each network is extracted, and consecutive networks are compared using a distance or a similarity matrix. Networks with large distances (small similarities) are deemed anomalous. Decomposition-based methods use matrix or tensor decomposition methods for network anomaly detection. The cluster-based event detection focuses on communities or clusters instead of the whole network. Window-based methods have a time window, where ``normal'' behaviour is modelled using the networks in the previous window. When a new network occurs, it is compared with the normal before labelling it as anomalous or not.

Due to overlapping research interests in statistics, mathematics, computer science and social science, multiple words are used to describe the same object. For example, the words `network' and `graph' are interchangeable. A vertex is often referred to as a node, and depending on context it can be called an actor; an edge is also called a link. The explosion in machine learning has produced a plethora of graph-based techniques applicable for different problem contexts. A recent survey by \citet{Ma2021} explores deep learning methodologies in network anomaly detection. Although statistical methods for network anomaly detection have received limited attention, there is a rich history of network modelling for inference using statistical tools. Exponential Random Graph Models (ERGM), introduced by \citet{Holland1981}, are a popular tool for network inference. ERGMs have been used in different contexts ranging from social media networks \citep{Robins2007} to economic applications \citep{Jackson2011}. \citet{Sharifnia2022} use an ERGM framework to detect changes in social networks. First they select a well-fitted ERGM for their data and estimate the coefficients for all configurations. Then they use Hotelling's $T^2$ distribution and Exponentially Weighted Moving Average (EWMA) control charts to monitor the coefficients and detect changes. \citet{Tsikerdekis2021} use ERGMs to detect attacks in computer networks.

We do not use ERGMs, largely because an ERGM needs to be well-fitted and a single set of ERGM covariates cannot explain networks from different contexts. Indeed, fitting an ERGM needs careful consideration of covariates so that model degeneracy is avoided. Both \citet{Sharifnia2022} and \citet{Tsikerdekis2021} have used ERGMs in a single context, which allows them to select the appropriate covariates or ERGM terms. We want our method to be usable for any sequence of dynamic networks in order to detect anomalies. To allow this level of generalizability, we use a feature-based, time series modelling technique to detect network anomalies. We call our method \textit{oddnet}. Oddnet uses time series analysis, which is equipped to deal with complex seasonal patterns that simple similarity measures cannot capture. We explain the details of oddnet in \autoref{methodology}. In \autoref{results} we evaluate the performance of oddnet on four synthetic experiments and four real datasets. The synthetic experiments consider different network generation models, and we compare the performance of oddnet with two other methods. The real datasets come from diverse sources. We analyse the output of oddnet and explore the networks visually in this section. In \autoref{conclusion} we provide some concluding remarks.

We have produced an R package \texttt{oddnet} \citep{oddnetR} containing our algorithm. We have also made all examples in this paper available at \url{https://github.com/sevvandi/supplementary_material/tree/master/oddnet}.

\section{Methodology}\label{methodology}

Consider a sequence of temporal networks/graphs $\{\mathcal{G}_t\}_{t = 1}^T$ where each network is a static graph $\mathcal{G}_t = (\mathcal{V}_t\, \mathcal{E}_t)$, with $\mathcal{V}_t$ denoting the vertices and $\mathcal{E}_t$ the edges. We are interested in the overall or long term process of network generation. An anomaly is defined as a point with low conditional probability
$$
	p_{t|t-1} = \mathcal{P}(\mathcal{G}_t \mid \mathcal{G}_1,\dots,\mathcal{G}_{t-1}).
$$
That is, a graph $\mathcal{G}_t$ is anomalous if its probability of occurrence is low given the graphs that have come before it.

While it is difficult to compute $p_{t|t-1}$ exactly, we can arrive at good approximations using different modelling techniques. Our approach models $p_{t|t-1}$ by transforming the graph to a feature space using a function $f$ and approximating the network generation process by modelling the features $f(\mathcal{G}_t)$ over time.

\subsection{Mapping from the graph space to the feature space}

Let $f$ denote a feature computation function such that $f: \mathscr{G} \to \mathbb{R}^n$ where $\mathscr{G}$ denotes the space of all networks, and let $\bm{x}_t$ denote the image of $\mathcal{G}_t$ under $f$. Thus $\bm{x}_t$ is a vector of $n$ features corresponding to $\mathcal{G}_t$:
\begin{equation}\label{eq:intro1}
	\bm{x}_t = f(\mathcal{G}_t) ,
\end{equation}
where $\bm{x}_t \in \mathbb{R}^n$. Let $\mathcal{F}$ denote the codomain of $f$, which we call the feature space. Consequently, we migrate from the network/graph space to the codomain --- a subset of Euclidean space. It is important that the feature space $\mathcal{F}$ and the graph space $\mathscr{G}$ have some agreement on the distribution of the graphs and their features. In other words, the feature space should retain some distributional aspects of the graph space. In order to analyse the relationship between the feature space and the graph space, we will define a further conditional probability
\begin{align*}
	\tilde{p}_{t|t-1} & = \mathcal{P}(\bm{x}_{t} \mid \bm{x}_1,\dots,\bm{x}_{t-1}).
\end{align*}
If there is strong distributional agreement between the feature space and graph space, then $p_{t|t-1}$ and $\tilde{p}_{t|t-1}$ should be similar. We define three levels of distributional agreement for a function $f$.

\begin{definition}\label{def:totalagreement}
	A function $f: \mathscr{G} \to \mathbb{R}^n$ satisfies the \textit{Total Agreement Condition} for a set of graphs $\{\mathcal{G}_t\}_{t = 1}^T$, if
	$$
		p_{t|t-1} = \tilde{p}_{t|t-1}\qquad\text{for all $t \in \{1, \ldots, T\}$.}
	$$
\end{definition}
The Total Agreement Condition is the strictest form of agreement between the two spaces $\mathscr{G}$ and $\mathcal{F}$. We define a somewhat relaxed version next.

\begin{definition}\label{def:approximateagreement}
	A function $f: \mathscr{G} \to \mathbb{R}^n$ satisfies the \textit{Approximate Agreement Condition} with parameter $\epsilon$, or said to be $\epsilon-$\textit{approximate} for a set of graphs $\{\mathcal{G}_t\}_{t = 1}^T$, if
	$$
		\left\vert p_{t|t-1} - \tilde{p}_{t|t-1} \right\vert \leq \epsilon\qquad \text{for all $t \in \{1, \ldots, T\}$,}
	$$
	where $\epsilon > 0$ and small.
\end{definition}
Unlike the Total Agreement Condition, the Approximate Agreement Condition requires the conditional probability of observing a graph $\mathcal{G}_t$ to be similar to observing $f(\mathcal{G}_t)$, which is an easier condition to meet. However, our interest is in anomalies, which are points with low $p_{t|t-1}$. For this reason we only need the graphs with low conditional probabilities to have similar probabilities in the feature space. We define this as follows:

\begin{definition}\label{def:anomalypreserving}
	A function $f: \mathscr{G} \to \mathbb{R}^n$ satisfies the \textit{Anomaly Preserving Condition} with parameter $\epsilon$, or said to be $\epsilon-$\textit{anomaly preserving} for a set of graphs $\{\mathcal{G}_t\}_{t = 1}^T$, if and only if $ p_{t|t-1} \leq \epsilon_1$ implies that there exists $ \epsilon_2 > 0$ with $\epsilon = \max(\epsilon_1, \epsilon_2)$ such that $\tilde{p}_{t|t-1} \leq \epsilon_2$. That is,
	$$
		p_{t|t-1} \leq \epsilon_1 \iff \tilde{p}_{t|t-1} \leq \epsilon_2 ,
	$$
	where both $\epsilon_1$ and $\epsilon_2$ are small.
\end{definition}
The anomaly preserving condition stipulates that anomalies in the graph space --- graphs with low conditional probability --- have low conditional probability in the feature space and vice versa. These three definitions are ordered in the sense that if a graph satisfies the Total Agreement Condition then it satisfies the Approximate Agreement Condition for all $\epsilon >0$. If a graph satisfies the Approximate Agreement Condition with parameter $\epsilon/2$ then it satisfies the Anomaly Preserving Condition with parameter $\epsilon$. We show this below.

\begin{lemma}\label{lemma:agreements}The agreement conditions are nested, i.e.,
	\begin{align*}
		\text{Total Agreement Condition}                               & \Rightarrow \text{Approximate Agreement Condition} \, ,        \\
		\text{ Approximate Agreement Condition}\left(\epsilon/2\right) & \Rightarrow \text{Anomaly Preserving Condition}(\epsilon) \, .
	\end{align*}
\end{lemma}
\begin{proof}
	Suppose a function $f$ satisfies the Total Agreement Condition with respect to a set of graphs $\{ \mathcal{G}_t\}_{t=1}^T$. Then for each graph $\mathcal{G}_t$ in the set, $p_{t|t-1} = \tilde{p}_{t|t-1}$, so the Approximate Agreement Condition is satisfied for all $\epsilon > 0$.

	If a function $f$ satisfies the Approximate Agreement Condition for a set of graphs $\{ \mathcal{G}_t\}_{t=1}^T$ with parameter $\epsilon/2$, then for all graphs $\mathcal{G}_t$ in the set we have $\left\vert p_{t|t-1} - \tilde{p}_{t|t-1} \right\vert \leq \epsilon/2$, so that $\tilde{p}_{t|t-1} \leq p_{t|t-1} + \epsilon/2$.

	Consider the subset of graphs $\mathcal{G}_t$ with $p_{t|t-1} \leq \epsilon_1 = \epsilon/2$. For these graphs we have
	$ \tilde{p}_{t|t-1} \leq \epsilon $ from the above inequality. Similarly, for the subset of points in the feature space with $\tilde{p}_{t|t-1} \leq \epsilon_2 = \epsilon/2$ we obtain $p_{t|t-1} \leq \epsilon$.
\end{proof}

As can be expected, things become simpler if the network generating process is iid. Then the Anomaly Preserving Condition implies that anomalies in the graph space are anomalies in the feature space without the need for conditioning.

\subsubsection{Examples}

We consider a couple of examples to illustrate these concepts. For both examples we consider a independent network generating process for ease of explanation. The first example looks at a sequence of 20 networks each having $N$ nodes where $N \in \{50, \ldots, 55\}$, generated using the Erdos-Renyi model with edge probability $p = 0.05$. Suppose an anomalous network is generated with an edge probability $p = 0.2$. Two normal networks and the anomalous network from this sequence are shown in \autoref{fig:ex1}.

\begin{figure}[!ht]
	\centering
	\includegraphics[scale = 0.6]{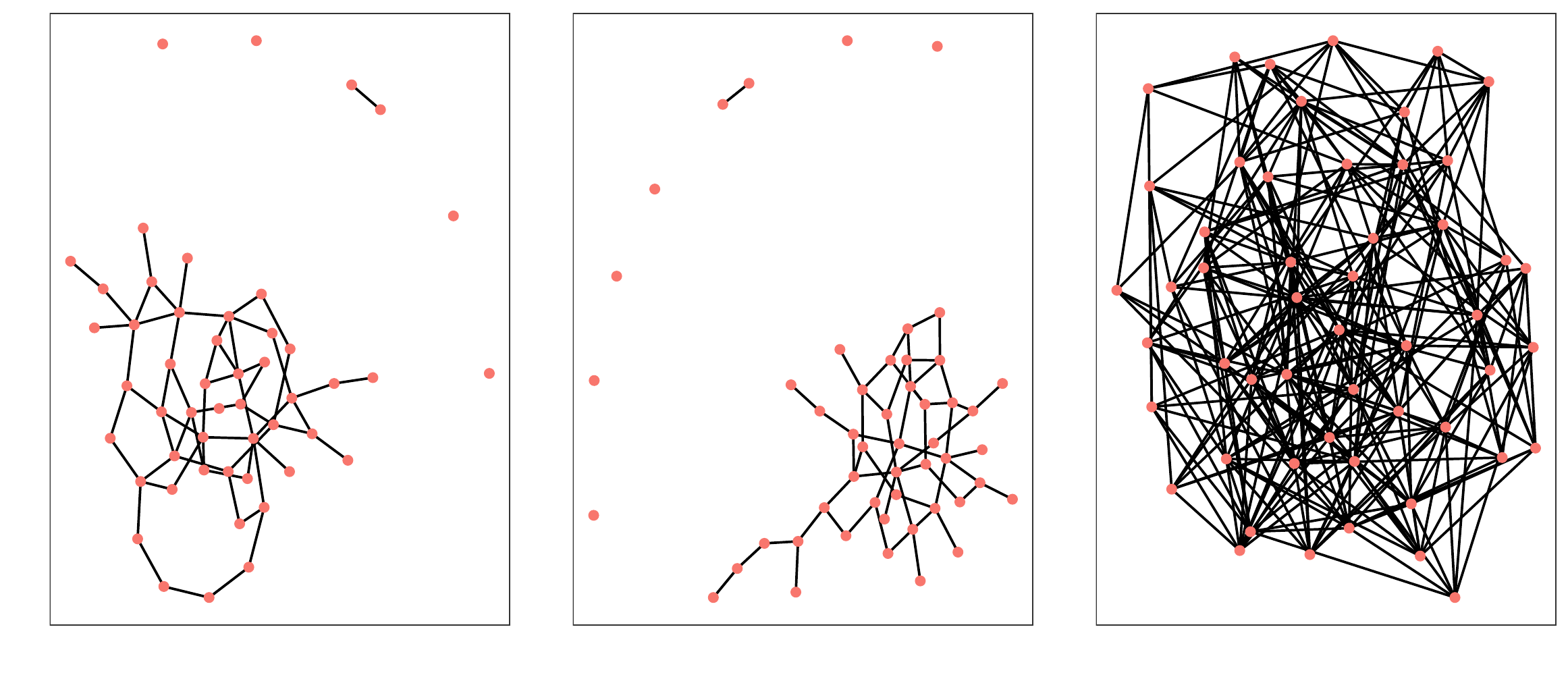}
	\caption{Three networks selected from a sequence of networks considered in Example 1. The first two are normal networks with edge probability 0.05 and the last shows an anomalous network having an edge probability 0.2. }
	\label{fig:ex1}
\end{figure}

As the second example we consider a sequence of 20 networks having $N \in \{50, \ldots, 55\}$ nodes with a total fixed number of edges randomly allocated. The anomalous network has all the edges emanating from a fixed node. Two normal networks from this sequence and the anomalous network are shown in \autoref{fig:ex2}.

\begin{figure}[!ht]
	\centering
	\includegraphics[scale = 0.6]{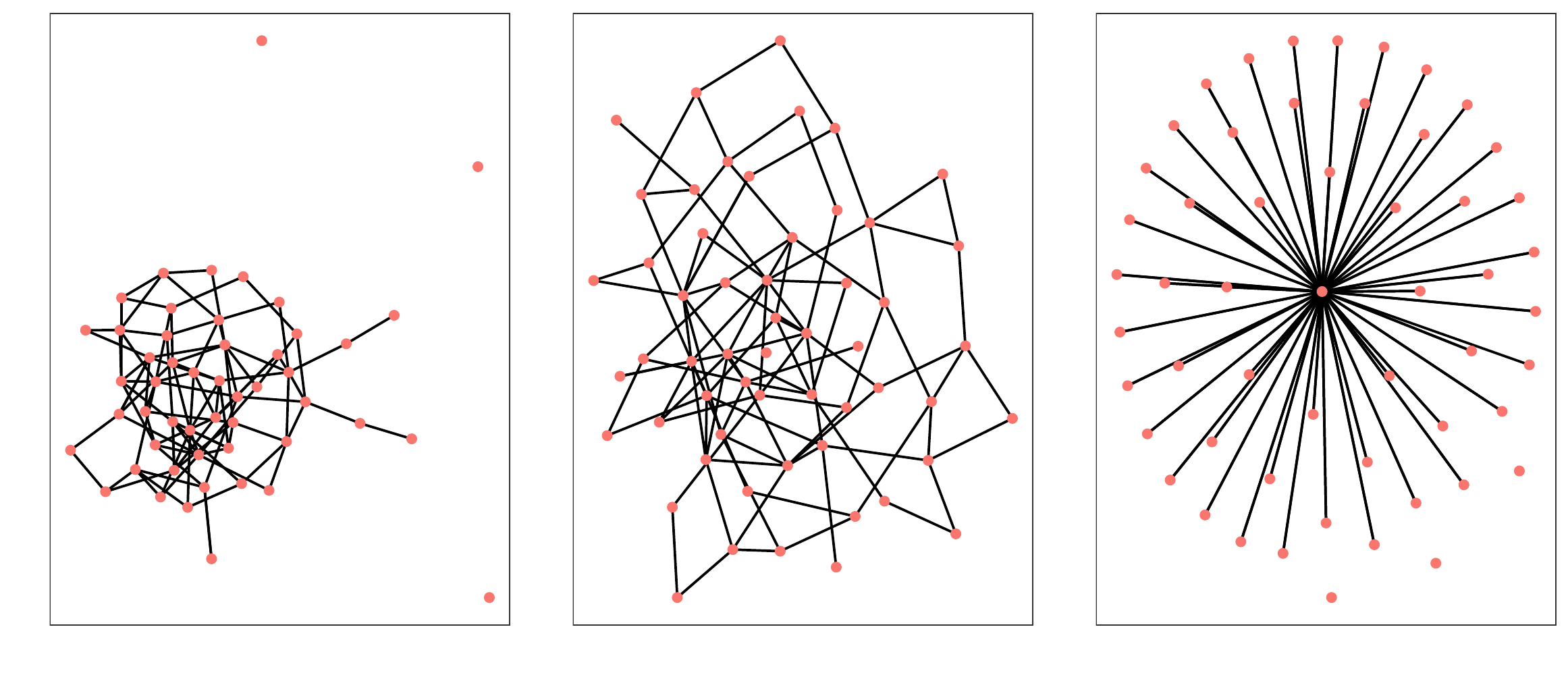}
	\caption{ Three networks shown from the sequence of networks in Example 2. Each network has 100 edges that are randomly selected. The anomalous network shown in the 3rd figure has all the edges connected to one common node. }
	\label{fig:ex2}
\end{figure}

For both examples, let us we consider the function $f = (f_1, f_2)$, where $f_1 =$ the number of edges and $f_2 =$ the number of nodes. In example 1, the number of edges in the anomalous network is much larger compared to the others. However, in example 2, the anomalous network has the same number of edges as the other networks. For both examples the number of nodes do not discriminate between the anomalous and the non-anomalous networks. The feature space for each example is shown in \autoref{fig:examplefeaturespace} with the anomalous network shown in a different colour. We see that the features can capture the anomaly in example 1, but not the anomaly in example 2. In example 1, the function $f$ satisfies the Anomaly Preserving Condition, i.e. the probability of anomalous networks in the feature space is low. However, in example 2, $f$ does not satisfy the Anomaly Preserving Condition, because the feature space is not rich enough to capture the particular anomaly. Incorporating features such as the clustering coefficient or the maximum node degree would solve the issue for this case. In practice, as we do not know the network generating process, we use a rich set of features.

\begin{figure}
	\centering
	\includegraphics[scale = 0.8]{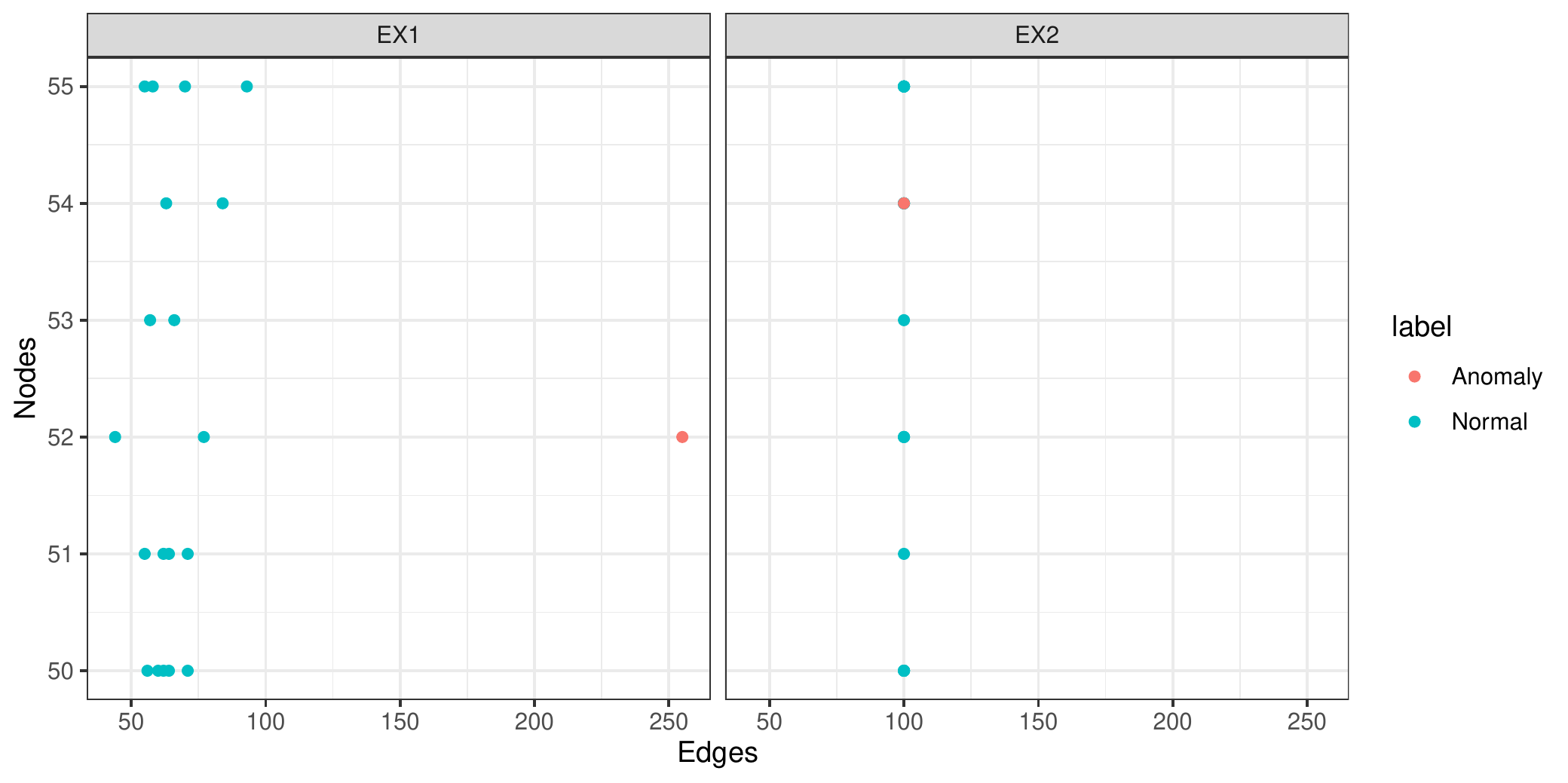}
	\caption{Feature space for examples 1 and 2.}
	\label{fig:examplefeaturespace}
\end{figure}

\subsection{The Features}\label{sec:features}

For each network $\mathcal{G}_t$ we compute 20 graph theoretic features from the R package \texttt{igraph} \citep{igraphR}; these are listed below.
\begin{compactenum}
	\item The number of nodes/vertices in $\mathcal{G}_t$.
	\item We compute the triangle distribution of $\mathcal{G}_t$ (the number of triangles connected to each node), and take the $99^{\text{th}}$ quantile of this distribution as the feature.
	\item We compute the degree distribution of $\mathcal{G}_t$ (the number of edges connected to each node), and take the $99^{\text{th}}$ quantile as our feature.
	\item The total number of edges in $\mathcal{G}_t$.
	\item The edge density is the ratio of the number of edges to the number of all possible edges.
	\item Transitivity, also known as friends of friends, or the clustering coefficient, measures the proportion of nodes where adjacent nodes are also connected. For example, if node A is connected to nodes B and C, then we consider if B and C are also connected.
	\item Assortativity takes an external property of the nodes into account. For example, consider the network of friends where an edge exists if A is friends with B. Suppose we add their political affiliations as attributes to each node. Now we can measure if friends have similar political preferences. This is known as homophily and sometimes referred to as `birds of a feather'. The assortativity coefficient measures the level of homophily in a graph. For graphs where nodes do not have an external attribute, the degree assortativity is computed.
	\item The mean graph distance calculates the mean of all shortest path distances between different nodes. If the graph is unconnected (i.e., not all nodes can be reached by a given node), then only the distances of the existing paths are considered.
	\item The diameter is the shortest distance between the two most distant nodes in a network.
	\item The proportion of isolated or non-connected nodes.
	\item The vertex connectivity gives the minimum number of vertices/nodes that needs to be removed to make the graph not strongly connected. A graph is a strongly connected if any vertex can be reached by any other vertex.
	\item The global efficiency is defined as the average inverse pairwise distances between all pairs of nodes.
	\item We extract two features from the connected components in the network. From the distribution of the number of nodes in each connected component we use the $99^{\text{th}}$ quantile as our feature.
	\item The number of connected components is also included.
	\item Centrality is a key aspect of network analysis. Closeness centrality of a vertex measures how close that vertex is to other vertices in the graph. It is defined as the inverse of the sum of distances to all vertices. We compute the closeness centrality for all vertices and include the proportion of vertices with closeness $ \geq 0.8$ in the feature vector. For this feature, we do not use the $99^{\text{th}}$ quantile because closeness centrality lies between 0 and 1 and for most graphs the $99^{\text{th}}$ quantile is equal to 1, making it a non-informative measure.
	\item Another important centrality measure is betweenness centrality. Betweenness measures how much a node connects other nodes by being a go-between. Suppose node A connects two groups of nodes, which otherwise would not be connected. In this case, node A has high betweenness centrality. It is defined as the number of shortest paths going through a node. Again, we compute the distribution for all nodes and take the $99^{\text{th}}$ quantile as our feature.
	\item PageRank is another measure of node importance. We compute the PageRank of all vertices and take the $99^{\text{th}}$ quantile.
	\item Hub scores also compute node importance. The hub scores of the vertices are defined as the principal eigenvector of $AA^T$ where $A$ denotes the adjacency matrix. We use the principal eigenvalue corresponding to the hub scores as a feature.
	\item Authority scores provide another measure of node importance. The authority scores of the vertices are the principal eigenvector of $A^TA$. We include the principal eigenvalue corresponding to the authority scores as a feature. The hub score and authority score features are identical for undirected graphs, but are different for directed graphs.
	\item Cores describe group/community aspects in a graph. The $k$-core of a graph is a maximal subgraph with minimum degree at least $k$. That is, each vertex in a $k$-core has degree greater than equal to $k$. For example, if a group of friends are in a $k$-core, each person in the group knows at least $k$ other people. The coreness of a vertex is defined as $k$ if it belongs to a $k$-core, not to a $(k+1)$-core. We compute the coreness for all vertices and include the $99^{\text{th}}$ quantile.
\end{compactenum}

For features where a distribution needs to be summarised, we have consistently used the $99^{\text{th}}$ quantile as our quantity of interest, as our purpose is anomaly detection, so we are interested in the tails of the probability distribution. In effect we strive to fulfil the Anomaly Preserving Condition by \begin{inparaenum}[(a)] \item including a diverse set of features; and \item using a large quantile where a feature distribution needs to be summarised. \end{inparaenum}

\subsection{Approximating the network generation process}

The network generating process can induce temporal dependencies in the feature space. For example, suppose the number of edges in a network increases over time. If at a particular time point there is a sudden drop in the number of edges, that would be considered an anomaly. We want to capture that behaviour. Finding anomalies in the feature space would not enable us to identify that particular anomaly as illustrated in \autoref{fig:nonfeatureanomaly}.

\begin{figure}[!ht]
	\centering
	\includegraphics[scale = 0.8]{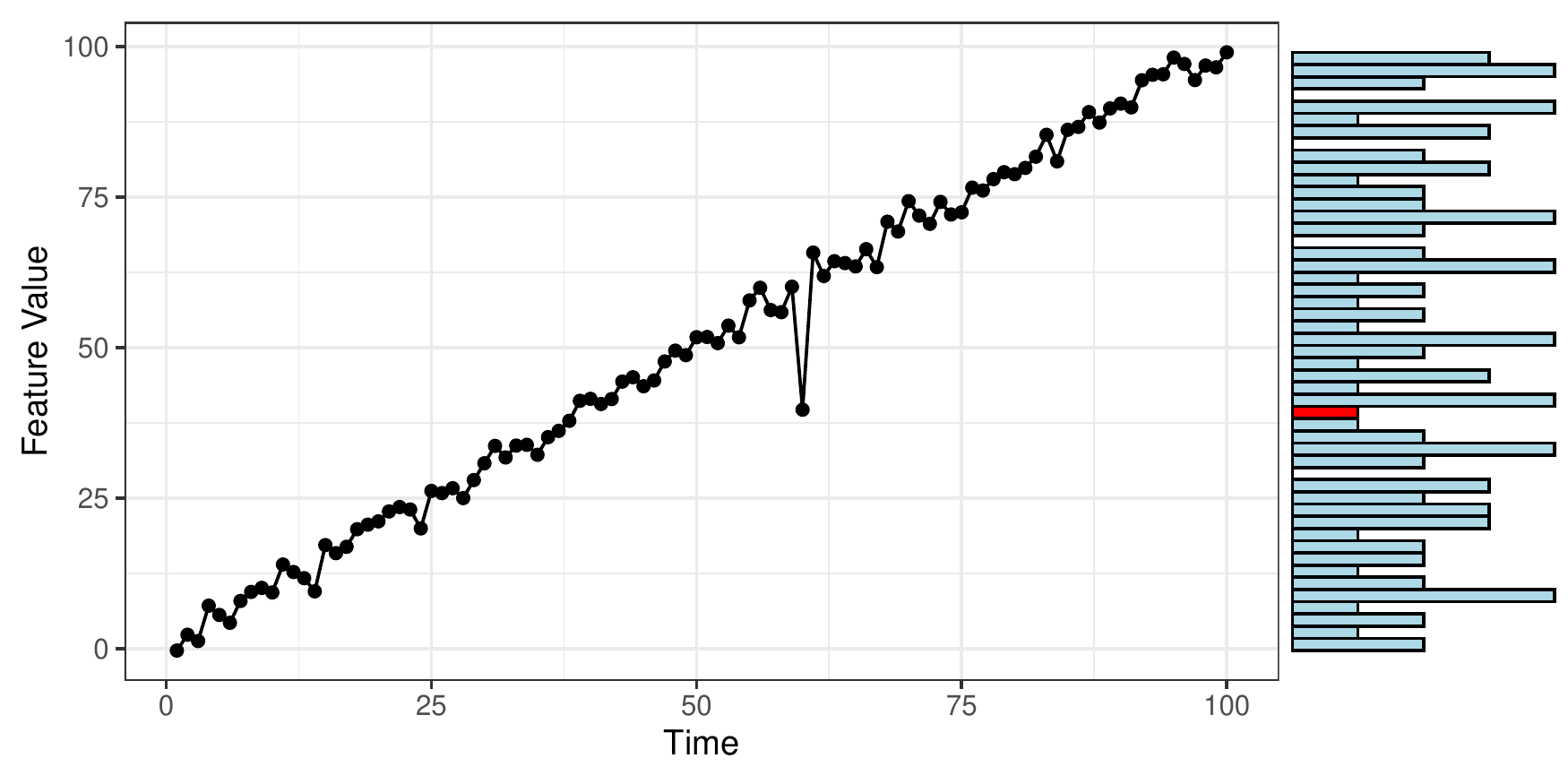}
	\caption{An anomaly takes place at time $= 60$. In the feature space, this observation is not anomalous as seen from the feature distribution on the y-axis. The bin corresponding to the anomalous observation is coloured in red.}
	\label{fig:nonfeatureanomaly}
\end{figure}

To model the temporal behaviour we use time series forecasting methods on the network features. Let $\bm{x}_t = (x_{t,1}, \dots,x_{t,n})$ where $f(\mathcal{G}_t) = \bm{x}_t$ and $x_{t,i}$ denotes the $i^{\text{th}}$ feature of $\mathcal{G}_t$. Let $\bm{x}_{ . ,i} = \{ x_{t, i} \}_{t = 1}^T $ denote the univariate time series of the $i^{\text{th}}$ feature. We fit ARIMA models to each feature time series $\bm{x}_{ . ,i}$ for $i \in \{1, \dots, n\}$ using automatic time series modelling methods \citep{hyndman2008automatic} via the R package \texttt{fable} \citep{fableR}. Thus, the best ARIMA model for each feature is fitted.

The residuals of the fitted ARIMA models are given by $e_{t,i} = x_{t, i} - \hat{x}_{t, i}$, where $\hat{x}_{t,i}$ is the forecast value of $x_{t,i}$ given $x_{1,i},\dots,x_{t-1,i}$ obtained from the corresponding ARIMA model. If the model has captured the temporal dependency adequately, then $e_{t,i} \sim \mathcal{NID}(0, \sigma^2)$ and the residuals will show no temporal dependence. Therefore, the transformation
\begin{equation}\label{eq:transformation}
	\mathcal{G}_t \mapsto \bm{x}_t \mapsto \bm{e}_t \, ,
\end{equation}
maps each network $\mathcal{G}_t$ to a point in $\bm{e}_t \in \mathbb{R}^n$ where they are free of temporal dependencies.

\begin{definition}\label{arimaapproximation}
	The ARIMA models \textit{approximate} the network generating process with parameter $\delta$, if
	$$ \tilde{p}_{t|t-1} \geq \delta \iff \lVert \bm{e}_t \rVert \leq c\left( \delta \right) \, ,
	$$
	for small $\delta$. This ensures that ARIMA models capture the general patterns including trend and seasonality in the feature space by mapping high density points in the conditional distribution $f(\mathcal{G}_t) \mid \mathcal{G}_1,\dots,\mathcal{G}_{t-1} $ to small residuals.
\end{definition}

We can now show that if certain conditions are satisfied, the anomalous graphs that give rise to large ARIMA residuals correspond to anomalies in the residual space.

\begin{proposition}\label{propanomalous}
	If a function $f$ satisfies the Anomaly Preserving Condition with parameter $\epsilon$ and if the ARIMA models approximate the network generating process with parameter $\delta$ with $\delta \geq \epsilon$, then an anomalous graph $\mathcal{G}_t$ gives rise to an anomaly in the ARIMA residuals space, i.e., there exists small $\xi$ such that $ \mathcal{P}\left( \bm{e}_t \right) \leq \xi$, where $\xi = \xi \left( \delta \right)$ and $\bm{e}_t$ is the residual corresponding to graph $\mathcal{G}_t$.
\end{proposition}
\begin{proof}
	For an anomalous graph $\mathcal{G}_t$, we have $\mathcal{P}\left( \mathcal{G}_t \mid \mathcal{G}_1,\dots,\mathcal{G}_{t-1} \right) \leq \epsilon'$ for some small $\epsilon'$. As $f$ satisfies the Anomaly Preserving Condition, $\tilde{p}_{t|t-1} \leq \epsilon$ for $\epsilon$ small. Because the ARIMA models approximate the network generating process we have
	$$ \tilde{p}_{t|t-1} \geq \delta \iff \lVert \bm{e}_t \rVert \leq c\left( \delta \right) \, ,
	$$
	giving us
	$$ \tilde{p}_{t|t-1} \leq \delta \iff \lVert \bm{e}_t \rVert \geq c\left( \delta \right) \, .
	$$
	The anomalous graphs are a subset of the above set because
	$$ \tilde{p}_{t|t-1} \leq \epsilon \Rightarrow \tilde{p}_{t|t-1} \leq \delta \, ,
	$$
	as $\epsilon \leq \delta$. This gives us
	$$ \tilde{p}_{t|t-1} \leq \epsilon \Rightarrow \lVert \bm{e}_t \rVert \geq c\left( \delta \right) \, .
	$$
	Thus, the anomalous graphs give rise to large residuals. As
	$$ e_{t,i} \sim \mathcal{N} \left( 0, \sigma^2 \right) 
	$$
	it ensures that $\bm{e}_t$ corresponding to anomalous $\mathcal{G}_t$ lies in low density regions making
	$$ \mathcal{P}\left( \lVert \bm{e}_t \rVert \geq c\left( \delta \right) \right) \leq \xi \, ,
	$$
	where $\xi = \xi \left( \delta \right)$.
\end{proof}

Therefore, this series of transformations maps an anomalous graph to an anomalous point in the ARIMA residual space, given the two conditions are satisfied. Automated ARIMA modelling is well recognised for capturing complex temporal patterns. We note that the Anomaly Preserving Condition is not very strict compared to Total Agreement and Approximate Agreement conditions. Furthermore, by computing a diverse set of features we expect the Anomaly Preserving Condition to be satisfied for most cases.

We illustrate this by computing the residuals for the example in \autoref{fig:nonfeatureanomaly}. \autoref{fig:residualspace} shows the residuals of an automatic ARIMA model fitted to the same feature. We clearly see the residual distribution shown in the histogram along the y-axis capture the anomalous behaviour at $t = 60$.

\begin{figure}[!ht]
	\centering
	\includegraphics[scale = 0.8]{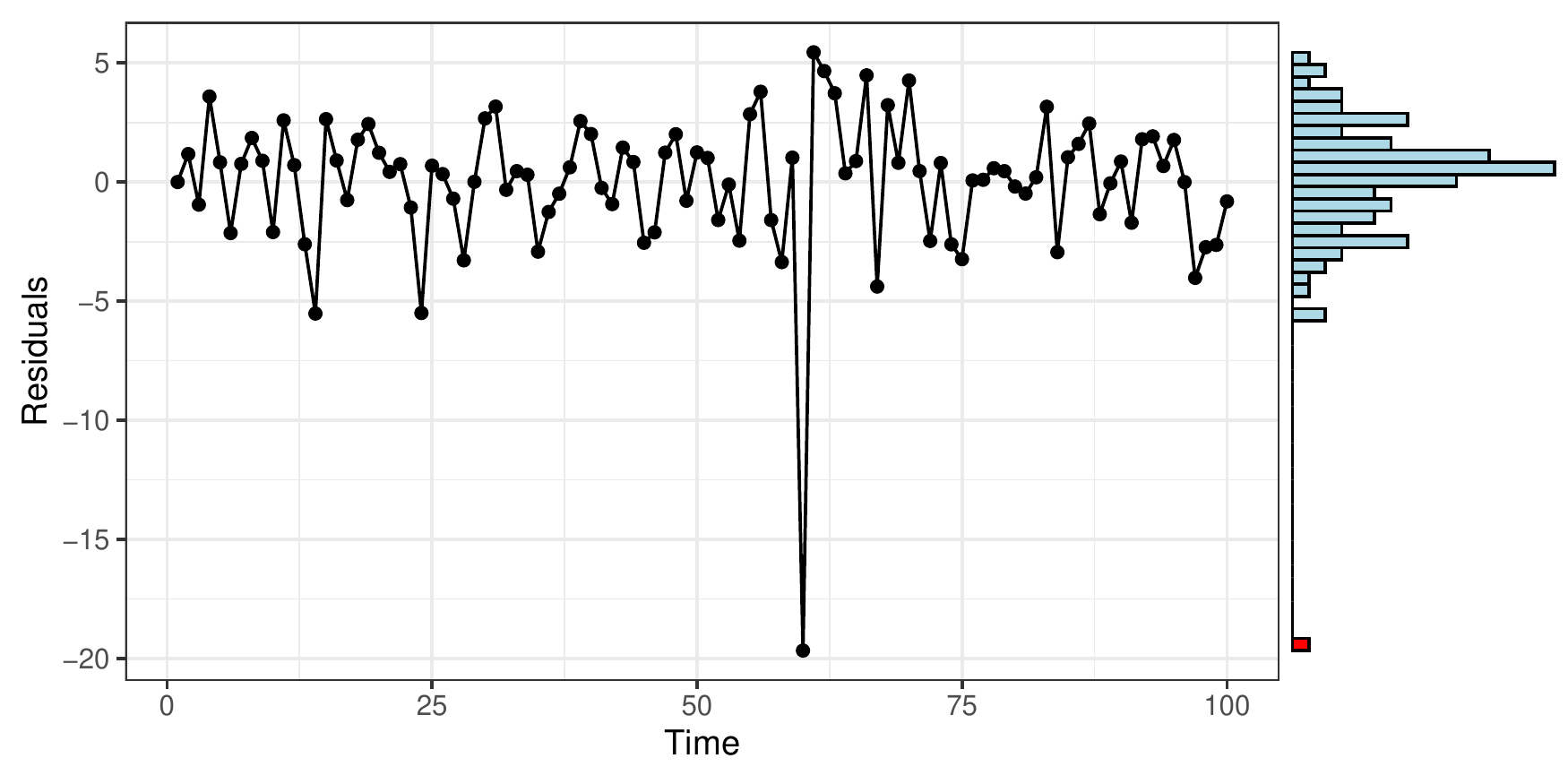}
	\caption{The residuals for the example in \autoref{fig:nonfeatureanomaly}. The histogram of the residual distribution is shown on the y-axis with the anomalous observation in red.}
	\label{fig:residualspace}
\end{figure}

\subsection{Dimension reduction and anomaly detection}\label{sec:dimredanomaly}

Finding anomalies in high dimensions is a challenge. Just as the curse of dimensionality affects other statistical learning tasks, it has an adverse effect on anomaly detection \citep{zimek2012survey}. It is hard to identify anomalies in high dimensions because points are far away from each other with many points residing in low density regions. 
Furthermore, errors of different features can be correlated. To aid this task we use dimension reduction. We experimented with two dimension reduction methods: dobin \citep{dobinpaper} and robust PCA \citep{croux2007algorithms}. For our network scenarios robust PCA gave better performance compared to dobin, possibly because dobin did not handle the correlations as well as robust PCA. Consequently, we selected robust PCA as the preferred dimension reduction method. Robust PCA uses a robust measure of variance for projecting the data into low dimensions. As the features and their residuals have different value ranges, we scale them using a trimmed mean and trimmed standard deviation as follows:
\begin{equation}\label{eq:scaling}
	\bm{y}_{.,i} = \frac{\bm{e}_{., i} - \tilde{\mu}_i}{ \tilde{s}_i} \, ,
\end{equation}
where $\bm{e}_{., i}$ denotes the residuals of the $i^{\text{th}}$ feature and $\tilde{\mu}_i$ and $\tilde{s}_i$ denote the trimmed mean and standard deviation of $\bm{e}_{., i}$ computed using the data within the 2.5 and 97.5 quantiles. We use robust PCA on the scaled residual space and use a 2-dimensional projection for anomaly detection.

We use the \textit{lookout} algorithm \citep{lookoutpaper}, a method that uses Extreme Value Theory (EVT) to detect anomalies. Anomaly detection methods using EVT have low false positives because of their inherent ability to handle fat tails. Lookout uses leave-one-out kernel density estimates for anomaly detection. A suitable bandwidth for anomaly detection is automatically selected using persistent homology, a technique in topological data analysis, bypassing the need to select a bandwidth manually.

\autoref{fig:methodology} shows different stages of our methodology. Algorithm \ref{algo:oddnet} sets out the key steps.

\begin{figure}[!ht]
	\centering
	\includegraphics[scale = 0.55]{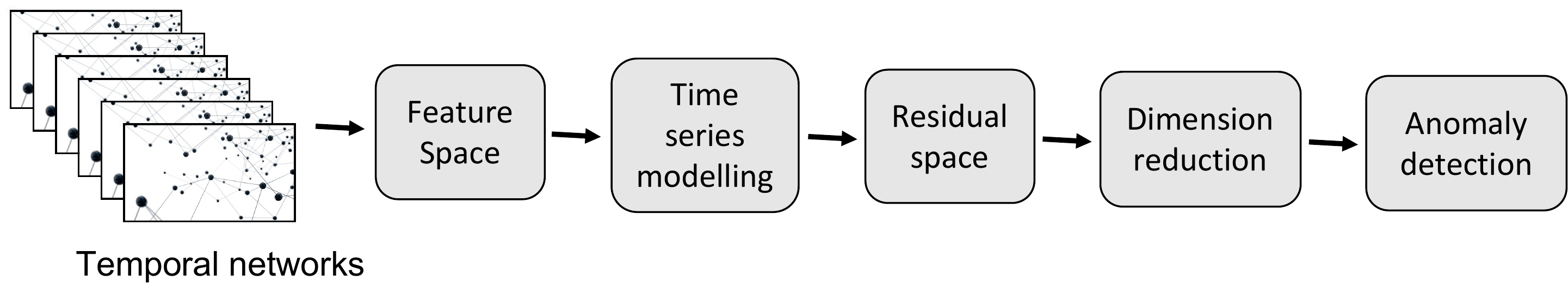}
	\caption{Different stages of oddnet.}
	\label{fig:methodology}
\end{figure}

\DontPrintSemicolon
\begin{algorithm}\fontsize{11}{12}\selectfont
	\SetKwInOut{Input}{input~~~}
	\SetKwInOut{Output}{output}
	\Input{~ The sequence of temporal networks $\left \{ \mathcal{G}_t \right\}_{t=1}^T$.}
	\Output{~ The anomalous networks}
	Map each network $\mathcal{G}_t$ to the feature space $\mathcal{F}$ by computing features $f\left( \mathcal{G}_t\right)$ as described in Section~\ref{sec:features}. Thus, $\mathcal{G}_t \rightarrow f\left( \mathcal{G}_t\right)$. \\
	Fit ARIMA models to $ \left \{ \bm{x}_t \right\}_{t=1}^T$ where $\bm{x}_t = f\left( \mathcal{G}_t\right)$. Let $\hat{\bm{x}}_t$ denote the fitted values of $\bm{x}_t$. \\
	Consider the residuals $\bm{e}_t = \bm{x}_t - \hat{\bm{x}}_t$ for $t \in \{1, \ldots, T\}$. \\
	Use robust PCA to reduce dimensions of the residual space. \\
	Use lookout to find anomalies in the 2-dimensional space.
	\caption{\itshape oddnet.}
	\label{algo:oddnet}

\end{algorithm}
\section{Results}\label{results}
In this section we conduct four experiments on synthetic networks and use four real datasets to identify anomalous networks. For the synthetic experiments we compare oddnet with two other methods: Laplacian Anomaly Detection (LAD) \citep{Huang2020} and Tensorsplat \citep{Koutra2012}. The real datasets do not have clearly defined, labelled anomalies, but we make use of prior analyses of the data sets when analysing the oddnet results, including where there is an existing analysis using LAD. Consequently, we use the real datasets to gain insights about the networks and possible interesting occurrences.

\subsection{Synthetic Experiments}\label{synthetic}

We generate networks using three models: Erd\H{o}s-R\'enyi random graph model \citep{Erdos1959}, Barab\'asi-Albert preferential attachment model \citep{albert2002statistical} and Watts-Strogatz small world model \citep{Watts1998}.

The Erd\H{o}s-R\'enyi model also known as the $G(n,p)$ model considers a graph with $n$ nodes with edge probability $p$ where an edge connects two nodes independently of other edges. Thus, the edges are allocated randomly. The total number of possible edges for a graph with $n$ nodes is ${n \choose 2}$. Using the $G(n,p)$ model, the probability of generating a graph with $M$ edges is $p^{M}(1-p)^{{n \choose 2} - M}$. Erd\H{o}s-R\'enyi graphs have been broadly used to solve combinatorial problems such as graph colouring. Due to its inherent randomness, Erd\H{o}s-R\'enyi graphs exhibit low clustering.

Unlike the Erd\H{o}s-R\'enyi model, the Barab\'asi-Albert (BA) preferential attachment model \citep{barabasi} advocates that new nodes are more likely to link with more connected nodes. These networks exhibit the `rich getting richer' phenomenon. They are also known as \textit{scale-free} networks. For a BA network, the probability of a new node connecting to an existing node $i$ is proportional to $k_i^{\alpha}$, where $k_i$ denotes the degree of node $i$. The exponent $\alpha$ is the key parameter of a BA model. When $\alpha = 0$ the BA model reduces to a random graph model as it is equally probable for a new node to connect to any other node. When $0 < \alpha < 1$, it is called sub-linear preferential attachment as the effect of preferential attachment is weak. Linear preferential attachment is given by $\alpha = 1$ and super-linear is given by $\alpha > 1$. It is when $\alpha \geq 1$ that the power of preferential attachment can be clearly observed. Preferential attachment models are commonly used to analyse social media networks as they show high levels of clustering.

The small world networks proposed by \citet{Watts1998} rewire regular networks, in which every node has constant degree, to introduce increasing amounts of disorder. The rewired networks tend to be highly clustered and yet have small path lengths. They start with a ring lattice with $n$ nodes and $k$ edges per node and rewire each edge randomly with probability $p$. Examples of small world networks can be found in neurological systems, power generator backbones and movie star networks. Figure~\ref{fig:networkgenmodels} shows three networks generated from these three network generating models. The differences between the networks are apparent.

For each synthetic experiment we focus on a single network generating model. The networks are generated for multiple values of the network generating parameter and for each parameter value, we use multiple randomisations. The anomalies are inserted at specific time points by changing the network generating parameters. We compare the performance of oddnet with Laplacian Anomaly Detection (LAD) \citep{Huang2020} and Tensorsplat \citep{Koutra2012} for the synthetic experiments. LAD computes the singular value decomposition of the Laplacian matrix for each network and uses the first $k$ singular values to identify anomalous behaviour. Tensorsplat performs PARAFAC decomposition and identifies anomalies using the low dimensional temporal factors. We use the Area Under the Receiver Operator Characteristic Curve (AUC) to compare the performance of the three methods.

\begin{figure}[!ht]
	\centering
	\includegraphics[scale = 0.8]{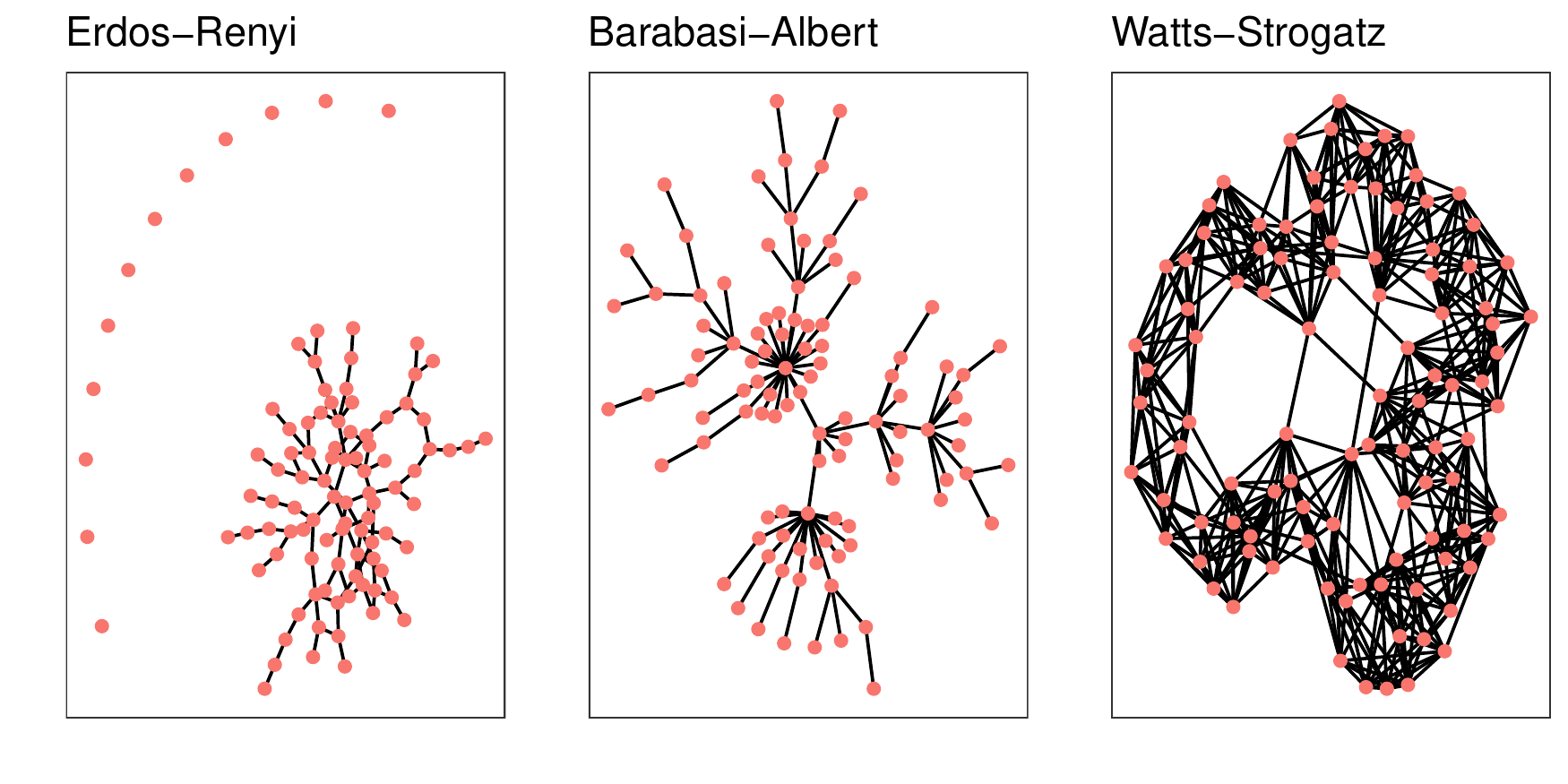}	\caption{Examples of networks generates from Erd\H{o}s-R\'enyi, Barab\'asi-Albert and Watts-Strogatz models.}
	\label{fig:networkgenmodels}
\end{figure}

\subsubsection{Experiment 1}

For the first experiment we use Erd\H{o}s-R\'enyi graphs with constant edge probability $p$; i.e., we model a static environment. We simulated a time series of 100 networks, each network having 100 nodes with $ p = 0.05$. The anomalous network occurred at time $t = 50$ and had edge probability $p = p_*$. We considered four iterations, each iteration having a unique value of $p_* \in \{0.1, 0.15, 0.2, 0.25 \}$. For each iteration we repeated the process 10 times to account for randomness.

\begin{figure}[!ht]
	\centering
	\includegraphics[scale = 0.8]{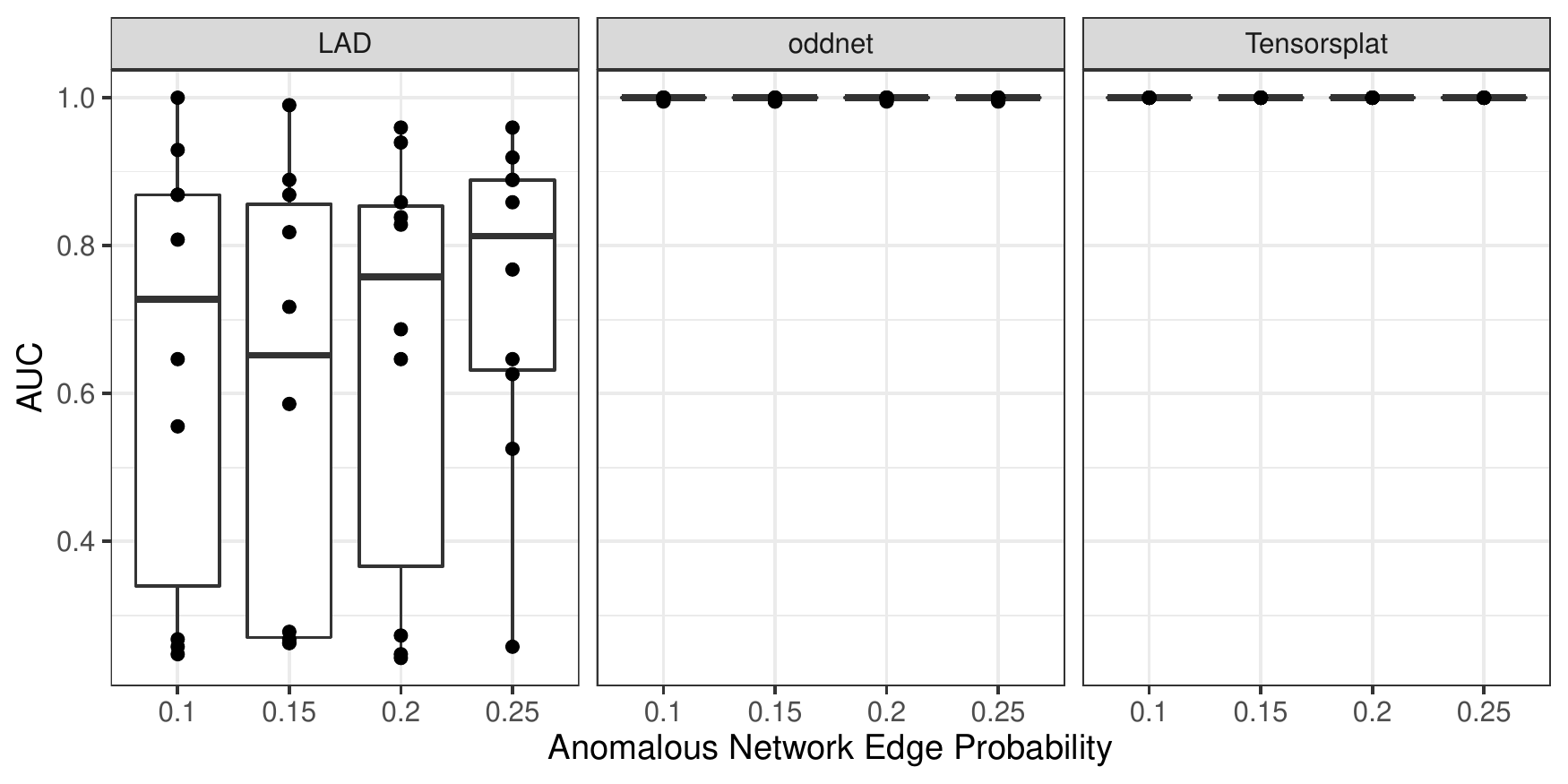}
	\caption{Experiment 1 results.}
	\label{fig:experiment1}
\end{figure}

\autoref{fig:experiment1} shows the results of the first experiment using boxplots. We see that both oddnet and Tensorsplat perform well for all parameter values. In contrast, LAD does not perform as well as the other two algorithms. From the viewpoint of the data, this is the easiest experiment because there are no temporal dependencies. The network generated at time $t_1$ is similar to the network generated at time $t_2$, which makes it easier to find anomalies.

\subsubsection{Experiment 2}

For the second experiment we simulated a dynamic environment with Erd\H{o}s-R\'enyi graphs. We constructed a series of 100 networks, each having 100 nodes with edge probability $p$ linearly changing from 0.05 to 0.5. An anomalous network is inserted at time $t = 50$ with $p = 0.2727 + p_*$ with $p_* \in \{0.05, 0.1, 0.15, 0.2 \}.$ As $p$ is linearly increasing from 0.05 to 0.5, the value 0.2727 corresponds to the edge probability of the $50^{\text{th}}$ network, before the anomaly is inserted. Allowing $p$ to change ensures that we have a dynamic sequence of networks, with the spiked $p_*$ at $t=50$ denoting the anomaly.

\begin{figure}[!ht]
	\centering
	\includegraphics[scale = 0.8]{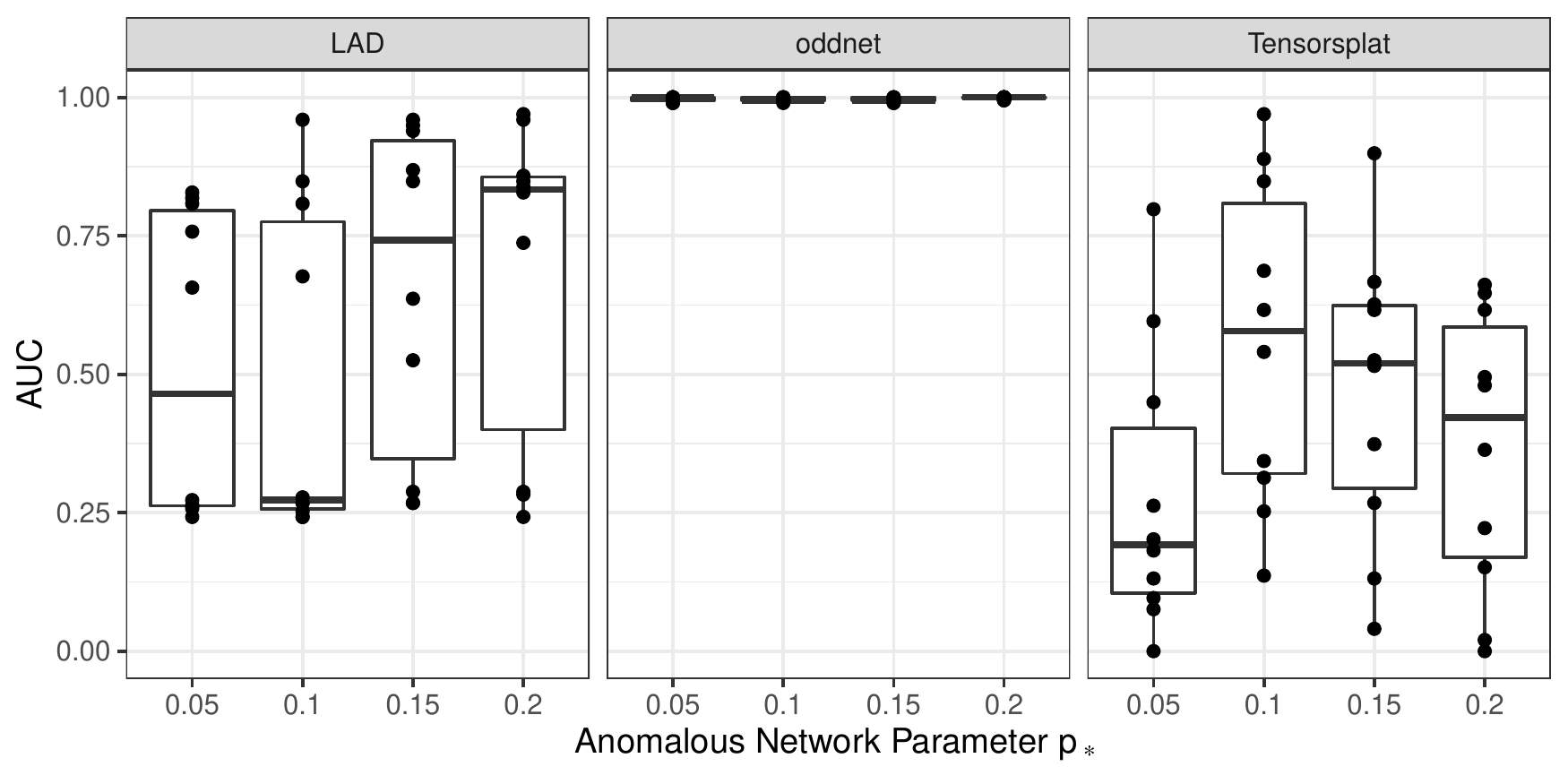}
	\caption{Experiment 2 results.}
	\label{fig:experiment2}
\end{figure}

\autoref{fig:experiment2} shows the results of the second experiment. We can see clearly that oddnet outperforms the other two methods. For this experiment LAD performs better than Tensorsplat.

\subsubsection{Experiment 3}

For the third experiment we used the preferential attachment model, which models the probability of a new node connecting to an existing node $i$ being proportional to $k_i^{\alpha}$, where $k_i$ is the degree of node $i$. We simulated a dynamic environment with $\alpha$ changing linearly from 1.1 to 1.9 over 100 networks. Each network has 100 nodes generated using this model. Again, an anomaly is inserted at $t = 50$ with $\alpha = 1.496 + p_*$ where $p_* \in \{0.25, 0.3, 0.35, 0.4 \}$. The value 1.496 corresponds to $\alpha$ of the $50^{\text{th}}$ network before inserting the anomaly. \autoref{fig:experiment3} shows the results, and we see that oddnet outperforms both LAD and Tensorsplat.

\begin{figure}[!ht]
	\centering
	\includegraphics[scale = 0.8]{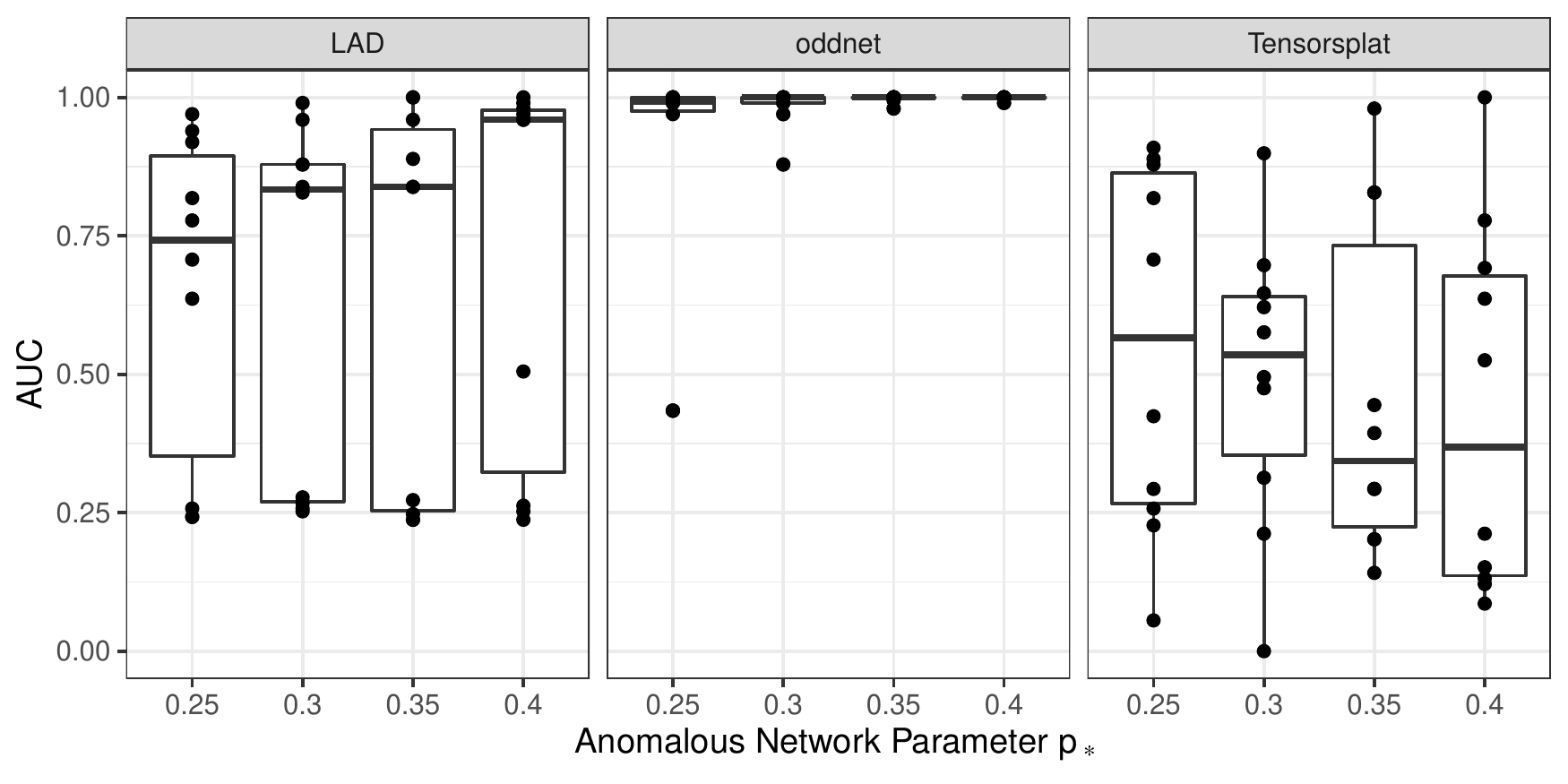}
	\caption{Experiment 3 results.}
	\label{fig:experiment3}
\end{figure}

\subsubsection{Experiment 4}

For the fourth experiment we used a small world model. We generated a series of 100 dynamic networks with the rewiring probability $p$ changing from 0.05 to 0.3. Each network is generated with 100 nodes. An anomaly is inserted at $t = 50$ with rewiring probability $p = 0.1737 + p_*$, with $p_* \in \{0.05, 0.1, 0.15, 0.2 \}$. As with previous examples, 0.1737 is the rewiring probability of the $50^{\text{th}}$ network before it was spiked to be anomalous.

\begin{figure}[!ht]
	\centering
	\includegraphics[scale = 0.8]{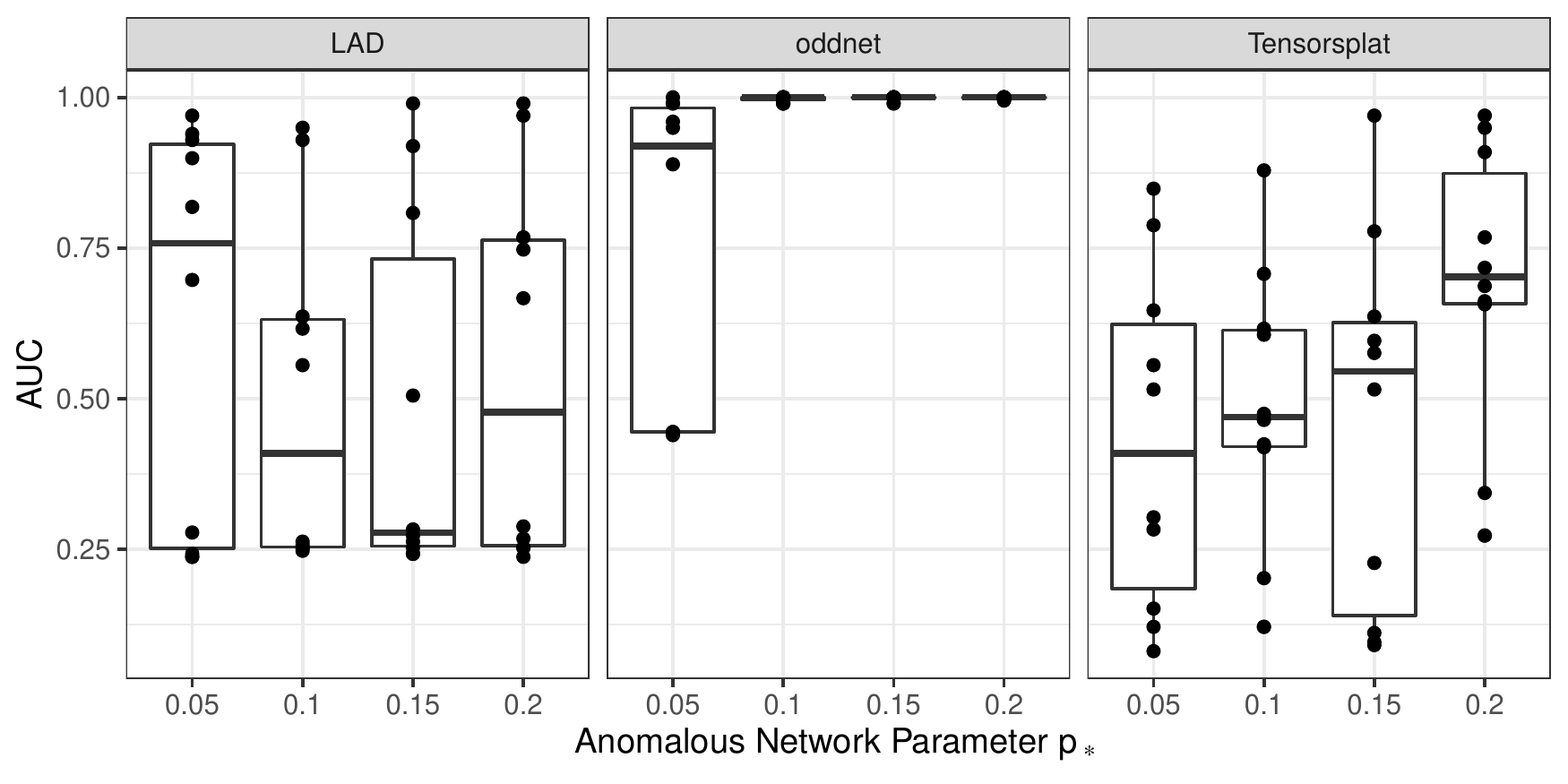}
	\caption{Experiment 4 results.}
	\label{fig:experiment4}
\end{figure}

\autoref{fig:experiment4} shows the results of this experiment. Again, we see that oddnet gives better performance than the other two methods.

\subsection{Real-world network datasets}

Next, we explore data arising from four real networks.

\subsubsection{US Senate co-voting network}

\citet{Lee2020} discuss and provide the data for the US Senate co-voting network from the $40^{\text{th}}$ Congress to the $113^{\text{th}}$ Congress. Their focus is on modelling the dynamic networks with models similar to Exponential Random Graph Models (ERGM) by accommodating temporal dynamics with varying coefficients. \citet{Huang2020} also discuss the US Senate co-voting network from the $93^{\text{th}}$ Congress to the $108^{\text{th}}$ Congress. Their focus is on anomaly detection. We use the networks provided by \citet{Lee2020} because there is more data.

In this example a node denotes a senator and an edge is formed between two senators if they vote on the same bill. \autoref{fig:ussenate1} shows the voting patterns of the $40^{\text{th}}$, $64^{\text{th}}$, $89^{\text{th}}$ and $113^{\text{th}}$ Congresses with red denoting Republicans and blue denoting Democrats. We see the voting patterns significantly change over time with some Congresses being more clustered within the political parties.

\begin{figure}[!b]
	\centering
	\includegraphics[scale = 0.8]{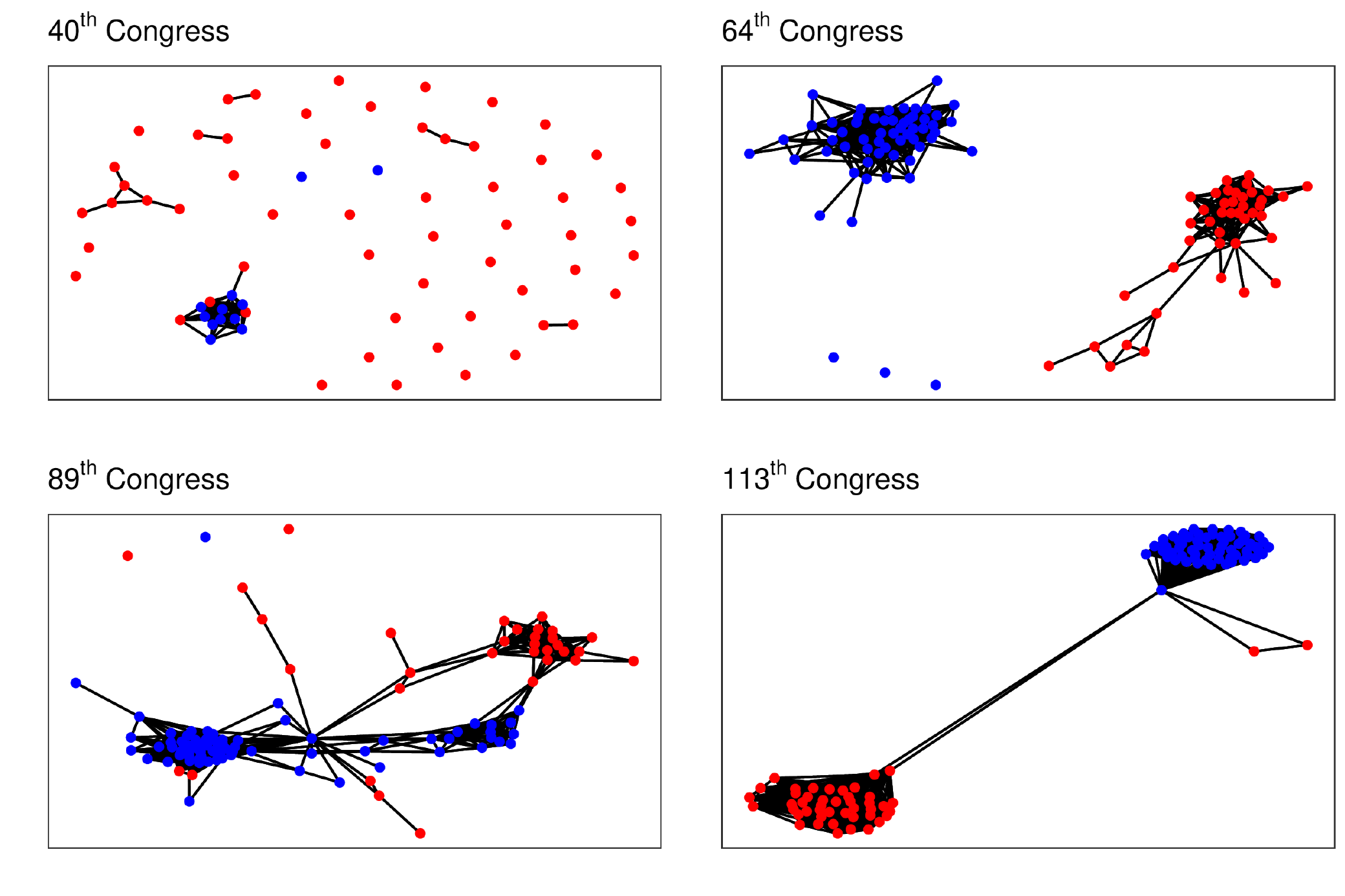}
	\caption{US Senate co-voting networks for the $40^{\text{th}}$, $64^{\text{th}}$, $89^{\text{th}}$ and $113^{\text{th}}$ Congresses.}
	\label{fig:ussenate1}
\end{figure}

\autoref{fig:ussenate2} shows the results of oddnet on this dataset. It shows the conditional probability of each network where networks with low conditional probabilities are considered anomalous. For all experiments, the level of significance $\alpha = 0.05$. It acts as a threshold and is shown by a dashed line. Points below the dashed line are identified as anomalies. The $100^{\text{th}}$ Congress is identified as an anomaly by oddnet. This result agrees with the work of \citet{Huang2020}; their largest anomaly is also the $100^{\text{th}}$ Congress, due to a relatively high level of collaboration.

\begin{figure}[!ht]
	\centering
	\includegraphics[scale = 0.8]{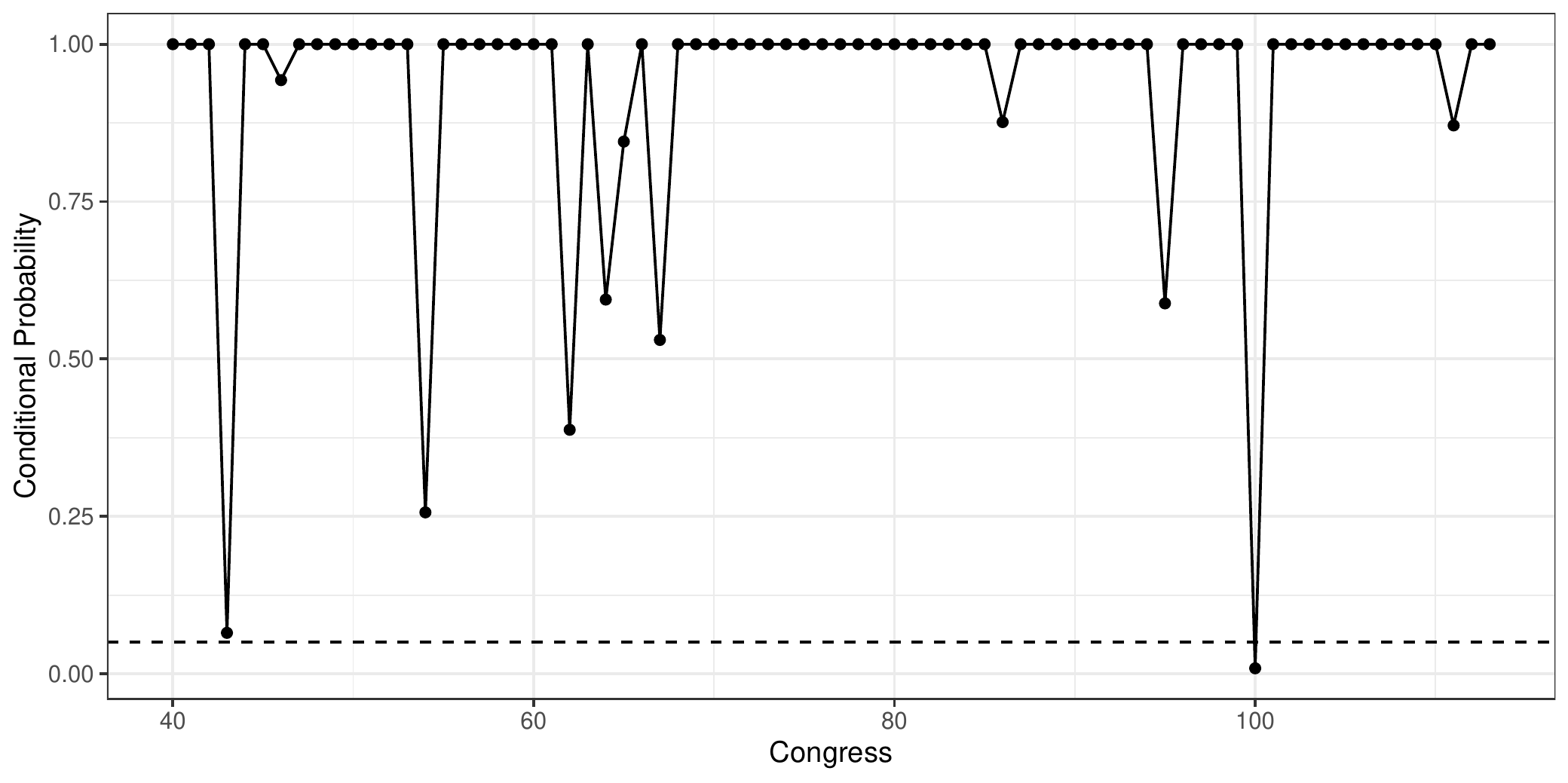}
	\caption{The conditional probability of each network calculated by oddnet for the US Senator co-voting networks.}
	\label{fig:ussenate2}
\end{figure}

Another Congress with low conditional probability is the $43^{\text{rd}}$ Congress, with a conditional probability of 0.06, which is just above the cut-off line $\alpha = 0.05$. As \citet{Huang2020} consider only networks from $93^{\text{th}}$ Congress to the $108^{\text{th}}$ Congress, their analysis does not include this network. The $43^{\text{rd}}$ Congress encompassed the years from 1873 to 1875. The \textit{Panic of 1873} caused an economic depression in the United States and was known as the \textit{Great Depression} until the 1929 event claimed the name \citep{barreyre_2011}. With widespread corruption, riots and public outcry, these years were marked with economic instability. \citet{barreyre_2011} explain how the ``economic crisis of 1873 transformed into a pivotal political event''. While not conclusive, this leads us to believe that the $43^{\text{rd}}$ Congress may have been quite different compared to others.

\subsubsection{UCI Message network}

\citet{Panzarasa2009} made available the network interactions of an online community of students at the University of California, Irvine. The dataset covers a period from April to October 2004. A total of 1899 users were recorded during this period. A student could use the online platform to send messages to other users; users could view and search the profile of others and send messages using the platform.

A directed edge is formed from user A to user B, if user A contacts user B via the platform. \autoref{fig:ucimessage1} shows the daily networks for six different days. We see enough evidence that it is not a static network. Of the networks shown, day 45 has a very high activity level compared to other days. \citet{Panzarasa2009} note two phases in this communication network. The first phase sees a rapid increase of users and acquaintances and lasts for approximately the first six weeks. They also list two additional dates corresponding to the end of the Spring term and the start of the Fall term: June 19 and September 20.

\begin{figure}[!ht]
	\centering
	\includegraphics[scale = 0.7]{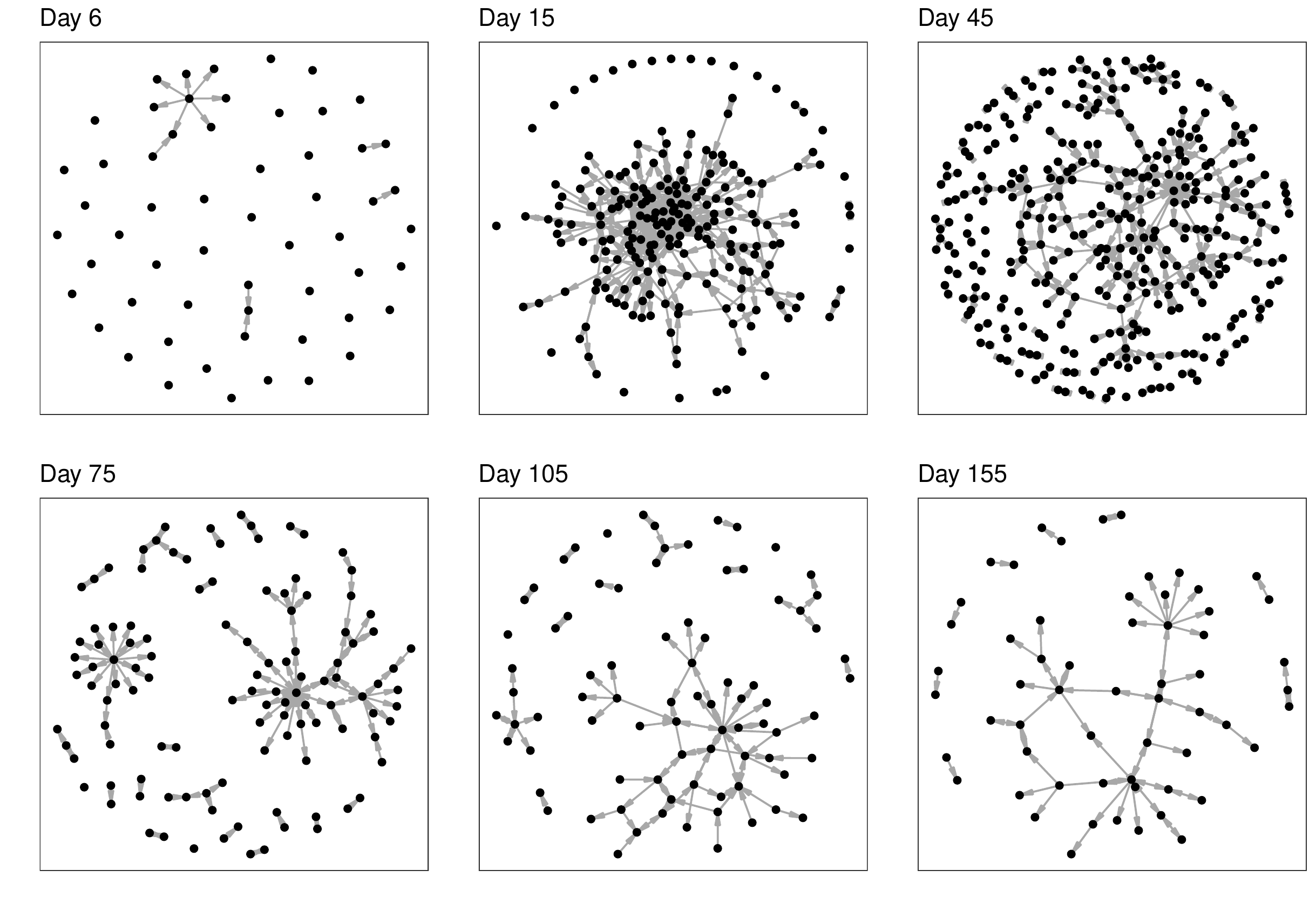}
	\caption{Days 6, 15, 45, 75, 105 and 155 of UCI message network.}
	\label{fig:ucimessage1}
\end{figure}

\autoref{fig:ucimessage2} shows the conditional probabilities of the networks with a dashed lined at $\alpha = 0.05$. We get four days flagged as anomalous. These are May 22, May 28, June 15 and September 24. May 28 is the $45^{\text{th}}$ day of the dataset and corresponds with the end of the first phase. Even though we do not get June 19 and September 20 flagged as anomalous days, we get two anomalies --- June 15 and September 24 --- that are close to the end of the Spring term and the start of the Fall term.

\begin{figure}[!tb]
	\centering
	\includegraphics[scale = 0.8]{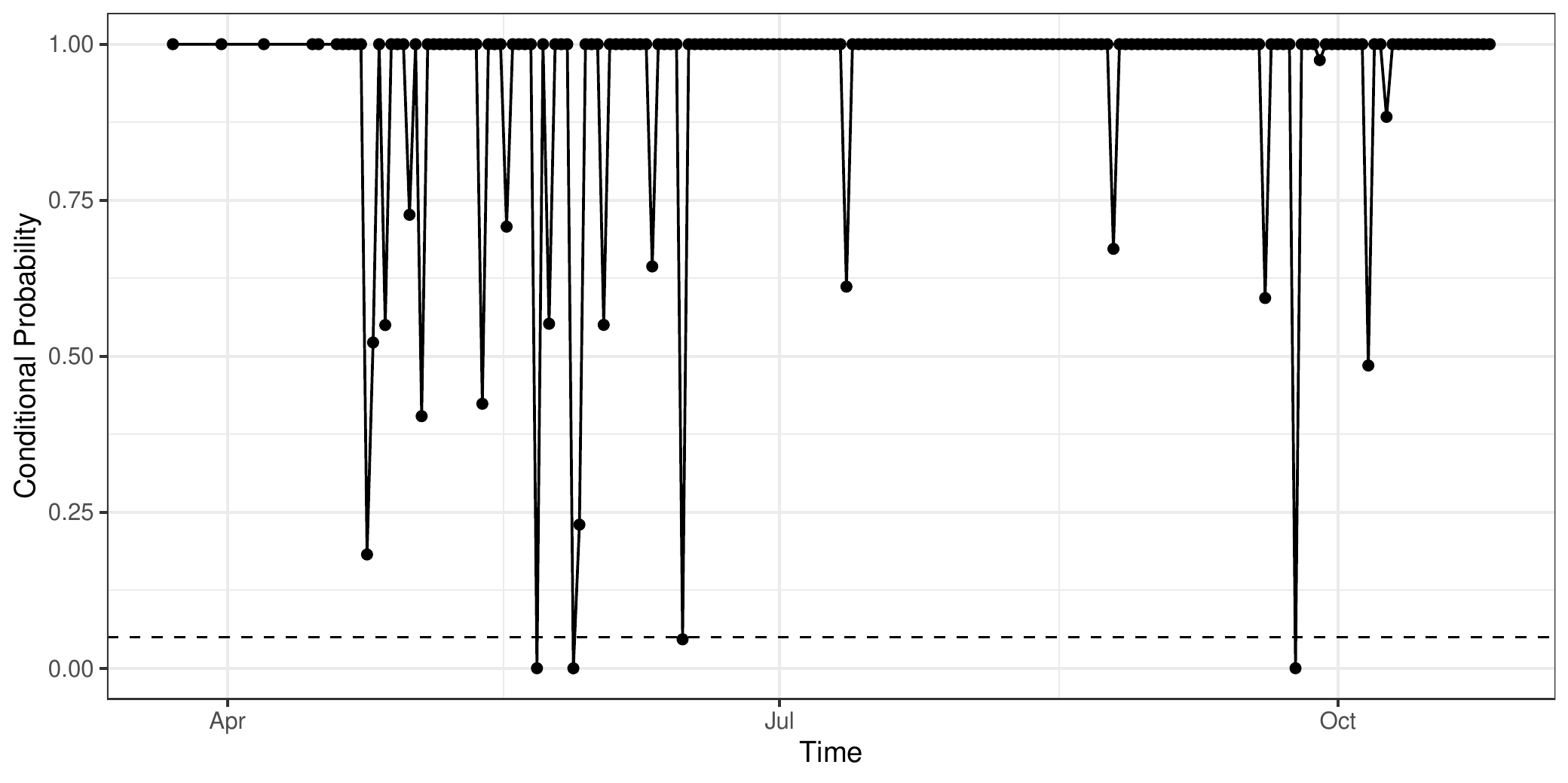}
	\caption{Oddnet conditional probabilities for the UCI message networks.}
	\label{fig:ucimessage2}
\end{figure}

\subsubsection{Canadian bill voting dataset}

\citet{Huang2020} have made available a dataset of the Canadian Parliament's bill voting patterns. The dataset spans from 2006 to 2019. They note that the House of Commons increased the number of electorates from 308 to 338 in 2015. They also note that 2015 is anomalous from another aspect: the Liberal party won an additional 148 seats taking up a total of 184 seats and formed a majority government led by Justin Trudeau. Prior to this, the Liberal party was divided. \citet{Cross2016} discuss the effects of a unified campaign at a local level.

\begin{figure}[!b]
	\centering
	\includegraphics[scale = 0.65]{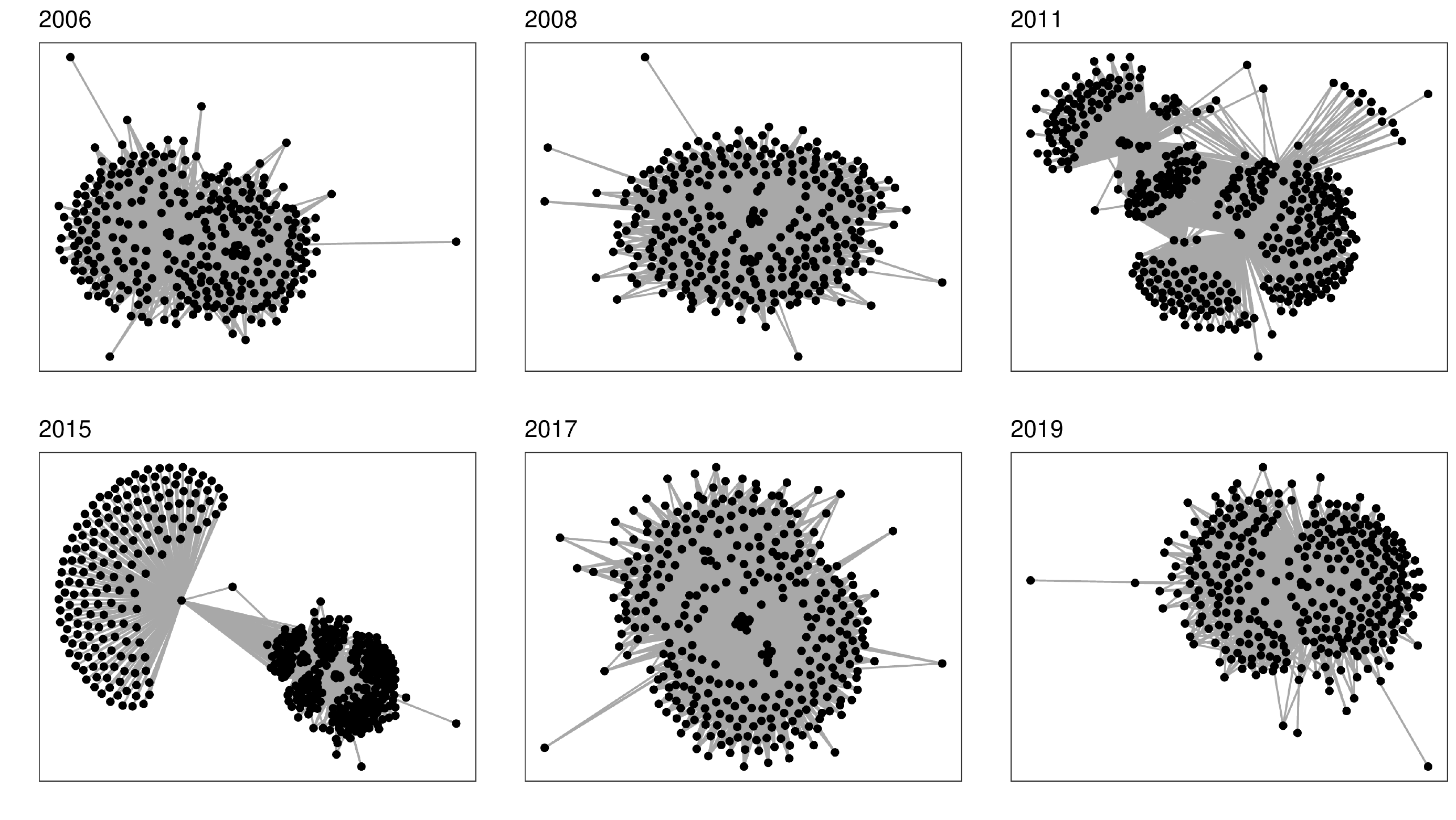}
	\caption{Canadian Parliament's bill voting networks for 2006, 2008, 2011, 2015, 2017 and 2019.}
	\label{fig:canadian1}
\end{figure}

\begin{figure}[!b]
	\centering
	\includegraphics[scale = 0.8]{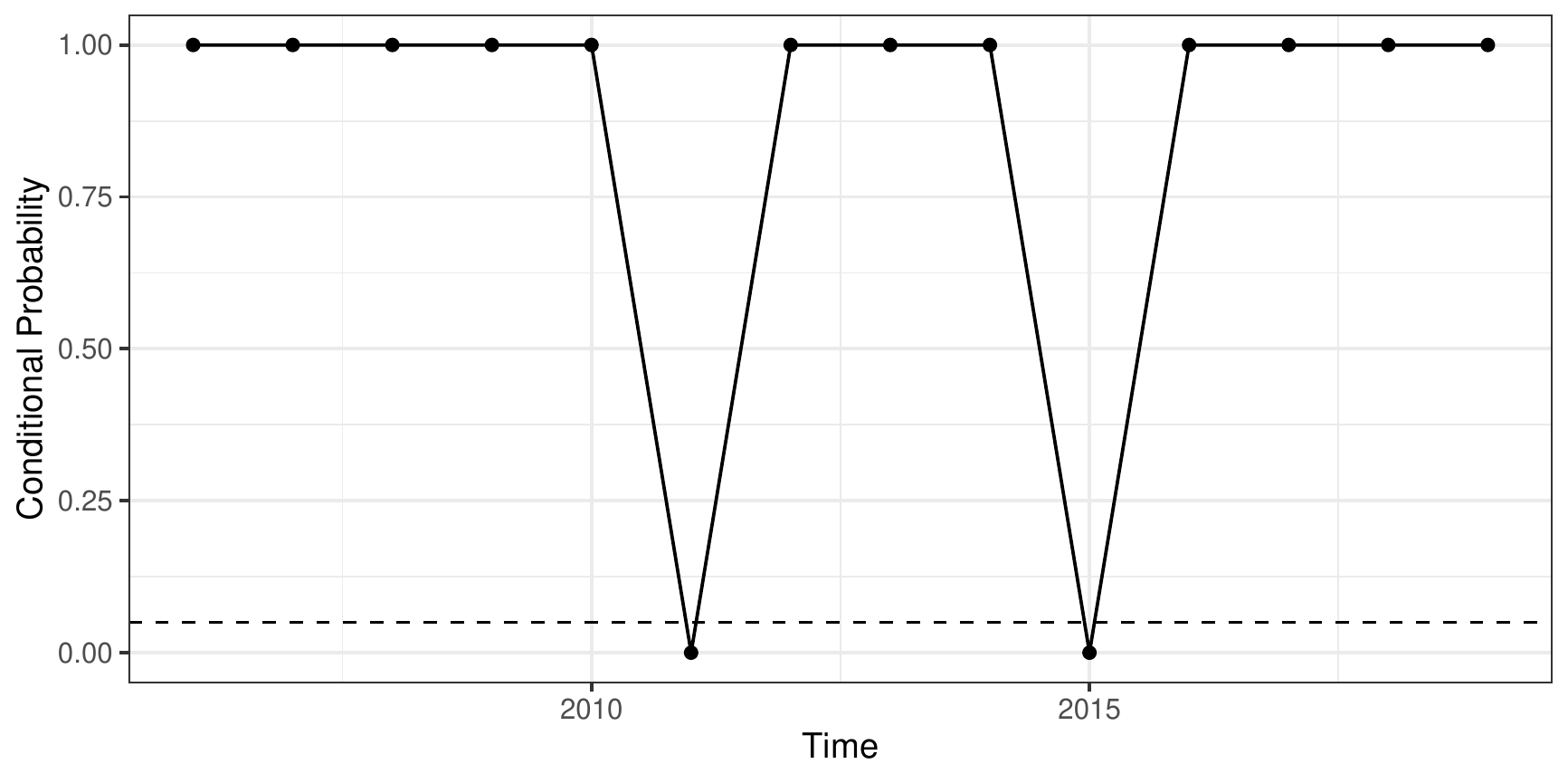}
	\caption{Oddnet conditional probabilities for the Canadian Parliament bill voting networks.}
	\label{fig:canadian2}
\end{figure}

\autoref{fig:canadian1} shows the networks for 2006, 2008, 2011, 2015, 2017 and 2019. The dataset does not include any party attributes, and so we have coloured all nodes in black. In 2015, we see clearly formed clusters and separation between two possible groups. The year 2011 is interesting because it resembles a network with multiple groups/factions.

Oddnet finds two anomalies from this dataset: 2011 and 2015. \autoref{fig:canadian2} shows the conditional probability of the networks with $\alpha = 0.05$ shown by a dashed line. The 2015 anomaly agrees with the work of \citet{Huang2020}. The anomaly in 2011 is somewhat different. In their book, \citet{gidengil2012dominance} ``explore the major fault lines that appeared in Canada's electoral landscape in the elections leading up to the 2011 electoral earthquake''. Therefore, it is fair to surmise that 2011 was a different year for Canadian politics. From the networks in \autoref{fig:canadian1} we see that both 2011 and 2015 are quite different from the others.

\subsubsection{US election blogs dataset}

\citet{Almquist2013} investigate the intra-group blog citation dynamics in the 2004 US presidential election collected by \citet{Butts2009}. The data was collected from July 22 to November 19 in 2004 at six hour intervals starting at midnight. There are 484 networks indexed by time. Each candidate's blog is considered a node and an edge exists from node $i$ to node $j$ at time $t$ if a link appears on blog $i$ to blog $j$ at time $t$.

\autoref{fig:usblogs1} shows blog networks at four different time stamps. Immediately we see that these networks are different from the previous examples in terms of network dynamics. Even though there are changes over time, the networks do not exhibit the same level of activity compared to previous examples. It is unreasonable to expect a blog post network to have similar dynamics to a social network.

\begin{figure}[!b]
	\centering
	\includegraphics[scale = 0.8]{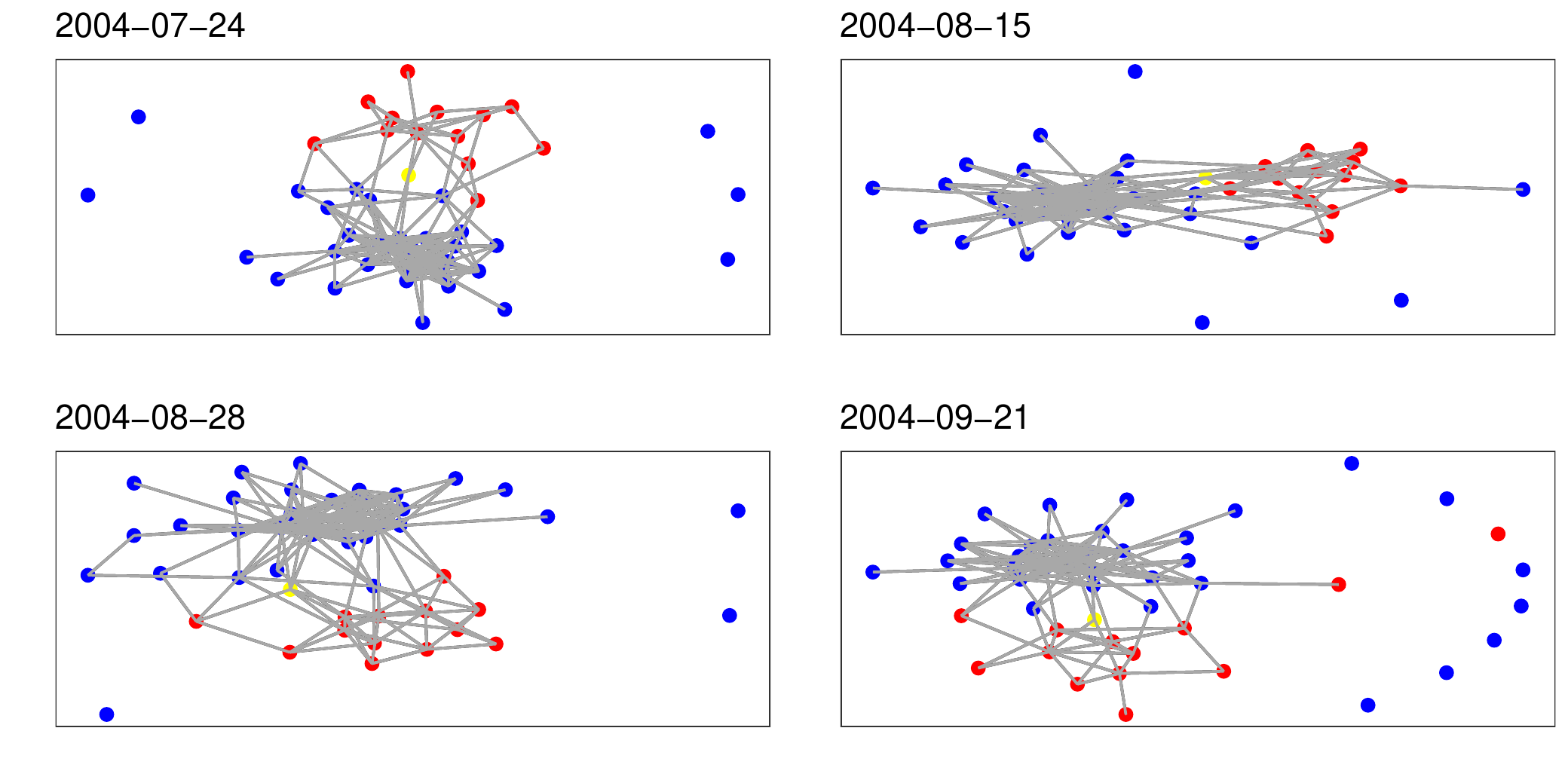}
	\caption{US election blog networks on four days with blue denoting Democratic blogs and red Republican blogs.}
	\label{fig:usblogs1}
\end{figure}

\begin{figure}[!tb]
	\centering
	\includegraphics[scale = 0.8]{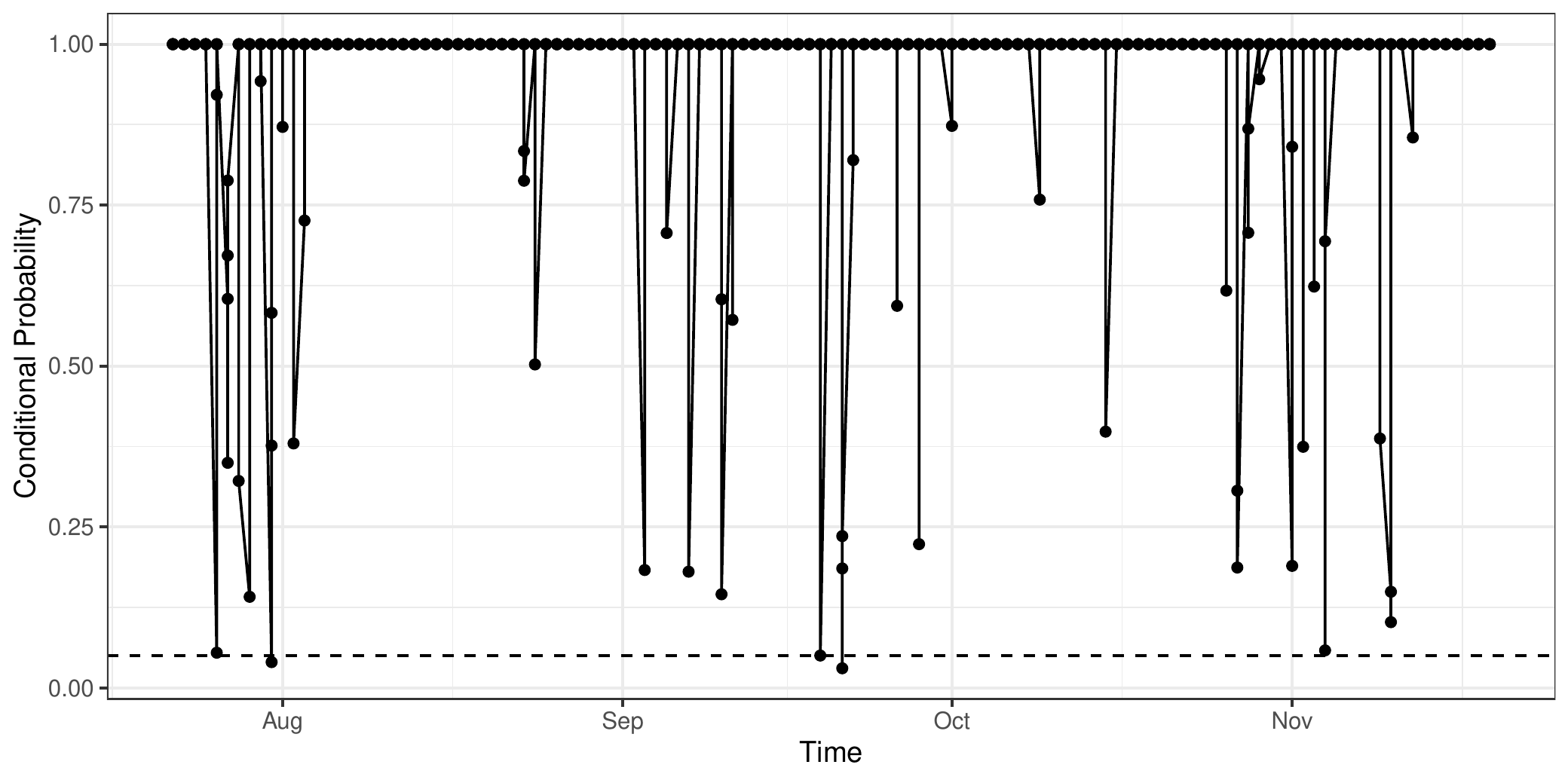}
	\caption{Oddnet conditional probabilities of the 2004 US election blog networks over time.}
	\label{fig:usblogs2}
\end{figure}

\autoref{fig:usblogs2} shows the conditional probability computed by oddnet with a horizontal dashed line at $\alpha = 0.05$. Oddnet identifies two anomalies corresponding to July 31 and September 21. These two anomalies correspond to two epochs identified by \citet{Butts2009} with a lag of one day. July 30 is noted as the end of the DNC convention and the start of RNC convention, and September 20 is the day the first presidential debate was held. Even though we do not identify these days as anomalous, we identify the next day as anomalous in both cases. We do not know if the blog posts were updated after the events took place. \autoref{fig:usblogs2} shows certain other dates with low conditional probability, even though they are not found anomalous. In addition to the anomalous dates, July 26 and November 4 also have low conditional probabilities. July 26 marks the start of the DNC Convention and November 3 marks the start of the post election period according to \citet{Butts2009}.

\section{Conclusion}\label{conclusion}

We have presented a statistical network anomaly detection method for dynamic/temporal networks. Oddnet computes network features and models them using time series methods. The residuals of the time series are used to find anomalies. The combination of time series modelling with network analysis is a major strength of oddnet. We demonstrate the effectiveness of oddnet using four synthetic experiments and four real examples. The synthetic experiments comprise sequences of networks generated using Erd\H{o}s-R\'enyi, Barab\'asi-Albert and Watts-Strogatz models. The results of oddnet on synthetic experiments are compared with LAD and Tensorsplat, two other network anomaly detection methods. Oddnet performs better than LAD and Tensorsplat. Oddnet's results on real examples pinpoint times corresponding to certain interesting occurrences. Future research avenues include extending oddnet to find anomalous subnetworks in large datasets.

\section{Supplementary materials}

The R package \texttt{oddnet} is available at \url{https://github.com/sevvandi/oddnet}. The programming scripts used in this paper are available at \url{https://github.com/sevvandi/supplementary_material/tree/master/oddnet}.

\bibliographystyle{agsm}
\bibliography{references.bib}

\end{document}